\documentclass[11pt]{article}
\usepackage{fullpage}
\usepackage[utf8]{inputenc}
\usepackage{algorithm}
\usepackage{algorithmicx}
\usepackage[noend]{algpseudocode}
\usepackage[toc,page]{appendix}
\usepackage{array}
\usepackage{authblk}
\usepackage{microtype}
\usepackage{tabularx}
\usepackage{booktabs}

\usepackage{amsmath,amssymb,amsthm}
\usepackage{verbatim}

\usepackage{boxedminipage}
\usepackage{calc}
\usepackage{hyperref}
\usepackage{mathtools}
\usepackage{refcount}
\usepackage{tikz}
\usepackage{thmtools}
\usepackage{thm-restate}
\usepackage{pgfplots}

\usepackage{todonotes}

\usetikzlibrary{automata, positioning, arrows}
\usetikzlibrary{shapes}
\usetikzlibrary{decorations.markings}
\tikzstyle{vertex}=[minimum size=2mm,circle,fill=black,inner sep=0mm,draw]
\tikzstyle{selected}=[minimum size=3mm,inner sep=0mm,draw,line width=1.5pt,path picture={% 
	\draw[red,fill=red] circle (1mm);
}]

\newcommand{\poly}{\text{\normalfont {poly}}}

\newcommand{\pbDefOpt}[3]{%
	\noindent
	\begin{center}
		\begin{boxedminipage}{0.99 \linewidth}
			{#1}
			\smallskip\\
			\begin{tabular}{lp{0.99 \textwidth - \widthof{~~~Output: }}}
				Input:&#2\\
				Output:&#3
			\end{tabular}
		\end{boxedminipage}
	\end{center}
}

\newcommand{\pbDefOptPara}[4]{%
	\noindent
	\begin{center}
		\begin{boxedminipage}{0.99 \linewidth}
			{#1}
			\smallskip\\
			\begin{tabular}{lp{0.99 \textwidth - \widthof{~~~Parameter: }}}
				Input:&#2\\
				Parameter:&#3\\
				Output:&#4
			\end{tabular}
		\end{boxedminipage}
	\end{center}
}

\newcommand{\pbDef}[4]{%
	\noindent
	\begin{center}
		\begin{boxedminipage}{0.99 \linewidth}
			{#1}
			\smallskip\\
			\begin{tabular}{lp{0.99 \textwidth - \widthof{~~~Parameter: }}}
				Input:&#2\\
				Parameter:&#3\\
				Question:&#4
			\end{tabular}
		\end{boxedminipage}
	\end{center}
}

\newcommand{\bitsize}{N}
\newcommand{\cF}{\mathcal{F}}

\DeclareMathOperator{\Und}{Und}
\DeclareMathOperator{\Def}{Def}
\DeclarePairedDelimiter\abs{\lvert}{\rvert}%
\DeclarePairedDelimiter\norm{\lVert}{\rVert}%
\DeclarePairedDelimiter\edge{(}{)}%

\makeatletter
\let\oldabs\abs
\def\abs{\@ifstar{\oldabs}{\oldabs*}}
\let\oldnorm\norm
\def\norm{\@ifstar{\oldnorm}{\oldnorm*}}
\makeatother

\newcommand{\ostar}{O^*}
\newcommand{\phivalue}{1.2321}

\newcommand{\appendixref}[1]{Appendix Section~\hyperref[#1]{\ref*{#1}}}

\newcommand{\branchon}[1]{\emph{2-way branch on whether to include $#1$}(\autoref{def:2waybranch_includev})}

\newtheorem{definition}{Definition}[section]
\newtheorem{lemma}{Lemma}
\newtheorem{theorem}{Theorem}
\newtheorem{hypothesis}{Hypothesis}
\newtheorem{simprule}{Simplification Rule}
\newtheorem*{simprule*}{Simplification Rule}
\newtheorem{simpruletwo}{Simplification Rule}

\theoremstyle{remark}
\title{Enumeration of Preferred Extensions in Almost Oriented Digraphs}

\author[1]{Serge Gaspers}
\author[2]{Ray Li}
\affil[1]{UNSW Sydney, Sydney, Australia and Data61, CSIRO, Australia }
\affil[1]{sergeg@cse.unsw.edu.au}
\affil[2]{UNSW Sydney, Sydney, Australia}
\affil[2]{rayli.main@gmail.com}

\date{}

\begin{document}
	
\maketitle

\begin{abstract}
    In this paper, we present enumeration algorithms to list all preferred extensions of an argumentation framework. This task is equivalent to enumerating all maximal semikernels of a directed graph. For directed graphs on $n$ vertices, all preferred extensions can be enumerated in $\ostar(3^{n/3})$ time and there are directed graphs with $\Omega(3^{n/3})$ preferred extensions.
    We give faster enumeration algorithms for directed graphs with at most $0.8004\cdot n$ vertices occurring in $2$-cycles.
    In particular, for oriented graphs (digraphs with no 2-cycles) one of our algorithms runs in time $O(\phivalue^n)$, and we show that there are oriented graphs with $\Omega(3^{n/6}) > \Omega(1.2009^n)$ preferred extensions.
	
	A combination of three algorithms leads to the fastest enumeration times for various proportions of the number of vertices in $2$-cycles. The most innovative one is a new 2-stage sampling algorithm, combined with a new parameterized enumeration algorithm, analyzed with a combination of the recent monotone local search technique (STOC 2016) and an extension thereof (ICALP 2017).
	%We also discuss lower bounds for both the number of preferred extensions and the running times of our algorithms.
\end{abstract}

\section{Introduction}

In Dung's theory of abstract argumentation \cite{dung}, an \emph{argumentation framework} (AF) is a digraph $G=(V,E)$, where each vertex represents an argument, and an arc $(u,v)\in E$ denotes that argument $u$ \emph{attacks} argument $v$. 
%The set of arcs is also called the \emph{attack relation}.
There are various semantics that express what properties a set of arguments should have for a rational agent to stand by that set of arguments.
One of the most central semantics is the \emph{preferred semantics} that was already proposed by Dung in his foundational paper \cite{dung}. Let $S\subseteq V$ be a subset of vertices (also called \emph{extension}) of a digraph $G=(V,E)$. The set $S$ is \emph{conflict-free} if no arc has both endpoints in $S$. A vertex $v\in V$ is \emph{acceptable} with respect to $S$ if %it is the case that
for each arc $(u,v)\in E$ there is an arc $(w,u)\in E$ with $w\in S$. In other words, for each argument $u$ that attacks $v$, there is an argument $w$ in $S$ that attacks $u$. We say in this case that $w$ \emph{defends} $v$ against $u$. The set $S$ is \emph{admissible} if it is conflict-free and each argument in $S$ is acceptable with respect to $S$. The set $S$ is \emph{preferred} if it is an inclusion-wise maximal admissible set.

While we will use the language of abstract argumentation, we remark that such vertex sets have also been studied in graph theory.
Neumann-Lara \cite{NeumannLara71} (see also \cite{Galeana-SanchezN84}) defined the notion of \emph{semikernels}. Maximal \emph{semikernels} are equal to the preferred extensions in the directed graph where all arcs are reversed.
The related notion of \emph{kernels} \cite{NeumannM44} has the same correspondence with stable extensions in abstract argumentation, and was introduced as an abstract solution concept in cooperative game theory, but has been extensively studied in the theory of directed graphs.
In particular, various issues around the enumeration of kernels and semikernels have been considered in previous work \cite{BanderierLR04,Bisdorff06,Galeana-SanchezL98,SzwarcfiterC94}.

\smallskip
\noindent
\textbf{Motivation}
%One of the 
A central problem in abstract argumentation is the enumeration of extensions prescribed by a given semantics.
In part, this is because exploring what sets of arguments may go together is an inherent issue of AFs.
The enumeration of preferred extensions is of particular interest, firstly for its own sake, but also in the study of other semantics as it forms the basis of several other semantics refining this set.
A number of existing algorithms and implementations enumerate all preferred extensions of a digraph (see, e.g., \cite{BistarelliRS15,Caminada07,CeruttiDGV13,CeruttiGV14,CeruttiVG18,CharwatDGWW15,preferredextensionsquery, ModgilC09,NofalAD14,VallatiCG14a,VallatiCG14}). The enumeration of preferred extensions is also a part of the biennial International Competition
on Computational Models of Argumentation (ICCMA).
Computational problems where the enumeration of extensions are used involve answering the questions: is a given argument in some / all preferred extensions and what is the number of preferred extensions containing a given argument / in total.
Upper bounds on the number of extensions under various semantics have also been proposed as fundamental characteristics to compare various semantics in abstract argumentation \cite{BaumannS14a,DunneDLW15}.

Dunne et al. \cite{DunneDLW15}, building on the work of Baumann and Strass \cite{BaumannS13}, showed that the number of preferred extensions is $O(3^{\abs{V}/3})$ (this result also holds for many other semantics \cite{BaumannS13, DunneDLW15}).
This bound is realized by a disjoint union of triangles, where every edge is replaced by an arc in both directions. 
The proof is based on the well-known Moon and Moser result \cite{MoonM65} for upper bounding the number of maximal cliques in graphs.
% Baumann and Strass \cite{BaumannS13} showed that the number of stable extensions (an extension $S$ is \emph{stable} if it is conflict-free and it attacks every vertex not in $S$) is $O(3^{n/3})$, where $n$ is the number of vertices. This bound is realized by a disjoint union of triangles, where every edge is replaced by an arc in both directions. 
% The proof is based on the well-known Moon and Moser result \cite{MoonM65} for upper bounding the number of maximal cliques in graphs.
% Based on this result, Dunne et al. \cite{DunneDLW15} proved that the same bound holds for several other semantics: semi-stable, naive, stage, and notably the preferred semantics.

We study the enumeration of preferred extensions in digraphs with no, or relatively few, $2$-cycles (i.e., bidirectional arcs). Our aim is to determine how much the presence of $2$-cycles affects the number of preferred extensions of an AF.
Mutually attacking arguments play a special role in abstract argumentation \cite{hierarchicalargumentation}, but this conflict is often resolved rather easily if the strength of the two attacks can be evaluated \cite{valuebasedafs}, or the user's preference between the two arguments can be elicited \cite{ModgilP12,preferencebasedafs}.
%The concept of a resolution of an AF was introduced by Modgil \cite{hierarchicalargumentation}. A \emph{resolution} of an AF $G=(V,E)$ is a digraph that can be obtained from $G$ by replacing each bidirectional arc by an arc in just one direction.
These methods of resolving conflicts motivate the study of problems, and in particular enumeration problems, for AFs with no or few 2-cycles.

\smallskip
\noindent
\textbf{Our results} We study enumeration algorithms and combinatorial upper bounds on the number of preferred extensions in \emph{oriented graphs}, which are digraphs without 2-cycles, and generalizations of oriented graphs.
Our main concern is the enumeration of all preferred extensions in time moderately exponential in the number of vertices $n$, and we mainly focus on digraphs that are either oriented or have small resolution order.
%For a (directed) graph $G$, a \emph{modulator} to a class of (directed) graphs $\mathcal{C}$ is a set of vertices $X$ whose deletion gives a graph $G-X$ in the class $\mathcal{C}$.
%For a directed graph $G=(V,E)$, we define a parameter, $r(G)$, as the size of a smallest modulator the the class of oriented graphs.
%We are not aware of previous algorithms whose running time depends on this parameter, nor, indeed any other studies involving this parameter.
%We observe that by a simple reduction to \textsc{Vertex Cover}, $r(G)$ can be computed in $O\left(\min\left(1.2738^{r(G)}, 1.996^n \right) \right)$ time \cite{ChenKX10,XiaoN17}, and the computation of $r(G)$ will not be a bottleneck in any of our algorithms.
%For this, we define the \emph{resolution graph} of a digraph $G=(V,E)$ as the digraph $G_r=(V_r,E_r)$ where $V_r$ and $E_r$ contain those vertices and arcs that belong to a 2-cycle in $G$. The \emph{resolution order} and the \emph{resolution size} of $G$ is $r(G) := |V_r|$ and $R(G) := |E_r|$, respectively. Note that $r(G)/2 \le R(G) \le \binom{r(G)}{2}$. We focus on $r(G)$, the number of vertices occurring in 2-cycles, as a graph parameter.
The \emph{resolution order} of a digraph $G=(V,E)$, denoted $r(G)$, is the number of vertices that belong to a 2-cycle in $G$.

Our main result is an algorithm that, for any $\varepsilon>0$, enumerates all preferred extensions of a digraph $G$ on $n$ vertices in time
\begin{align*}
 &\ostar\left(\left(\min\left(\varphi^{2r}\cdot\varphi^{1-r},\: \left(\left(1+2^{\frac{1}{4}}-\frac{1}{\sqrt{2}}\right)^r\cdot \left(2-\frac{1}{\sqrt{2}}\right)^{1-r}\right)^{1+\varepsilon},\: 3^{1/3} \right) \right)^n \right)\\
 \le \; & \ostar\left(\left(\min\left(1.5180^{r}\cdot 1.2321^{1-r},\: 1.4822^r\cdot 1.2929^{1-r},\: 1.4423 \right) \right)^n \right),
\end{align*}
where $r=r(G)/n$, and $\varphi\approx 1.2321$ is the positive root of $1 - x^{-1} - x^{-8}$.
The $\ostar$ notation hides factors that are polynomial in the input size.
See \autoref{fig:run-time}, which plots the base $\alpha$ of the running time expressed as $\ostar(\alpha^n)$ for $r$ varying from 0 to 1.
For $r=1$, this is best possible and follows from the work in \cite{DunneDLW15,ModgilP12}.
\begin{figure}[tb]
	\begin{tikzpicture}
	\begin{axis}[axis lines=left,width=\textwidth,height=8cm,xmax=1,xmin=0,ymax=1.5, ymin=1.2, samples=50, xlabel=$r$, ylabel=$\alpha$]
	\addplot[blue, ultra thick] {pow(2-pow(2,-1/2),1-x)*pow(1+pow(2,1/4)-pow(2,-1/2),x)};
	\addplot[orange!70!yellow,  ultra thick] {pow(3,1/3)};
	\addplot[black!60!green, ultra thick] {pow(1.232054631,2*x)*pow(1.232054631,1-x)};
	\addplot[red, loosely dashed, very thick, fill=red, fill opacity=0.1] {pow(3,(1/3)*x)*pow(3,(1-x)*(1/6))} \closedcycle;
	\end{axis}
	\end{tikzpicture}
	\caption{\label{fig:run-time} The graph depicts the base of the exponential running times $\ostar(\alpha^n)$ of the three enumeration algorithms, according to $r=r(G)/n$. When $r< 0.6684$, our branching algorithm for oriented graphs together with the Oriented Translation (\textcolor{black!60!green}{dark green}) gives the fastest algorithm. For $r>0.8005$, the algorithm based on previous work \cite{DunneDLW15,ModgilP12} (\textcolor{orange}{orange}) is fastest. In the middle range, the combination of the 2-phase monotone local search with the parameterized enumeration algorithm (\textcolor{blue}{blue}) is fastest. Our lower bound on the largest number of preferred extensions is drawn with a \textcolor{red}{dashed red line}.}
\end{figure}
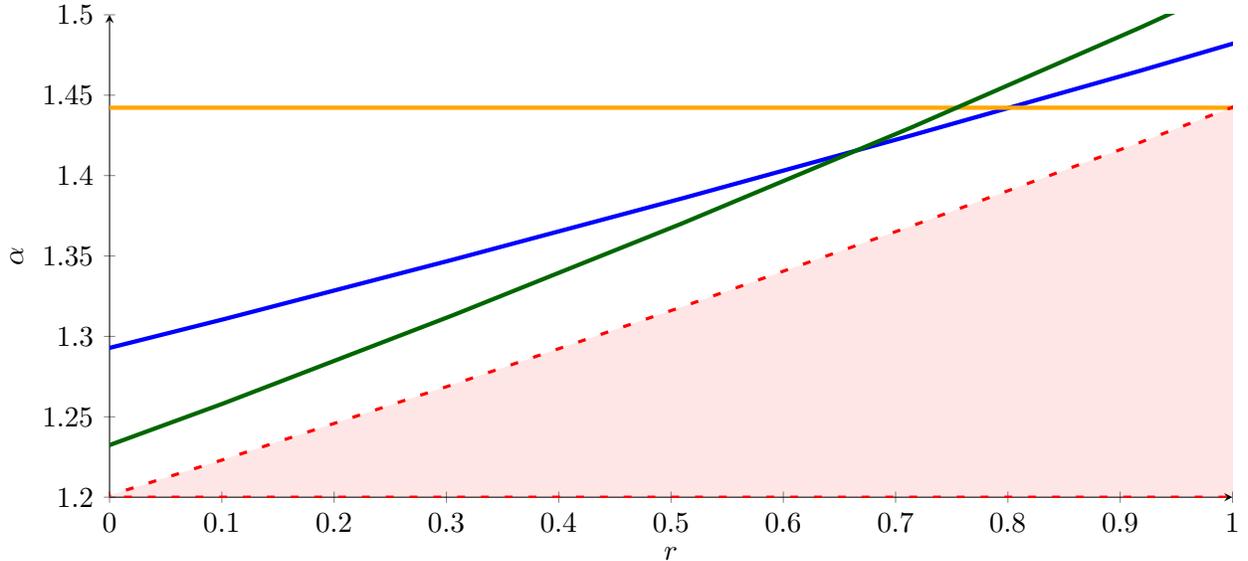
At the other end of the spectrum, i.e., for oriented graphs where $r=0$, the upper bound is $\ostar(\varphi^n) \le O(1.2321^n)$ and is obtained via a carefully constructed branching algorithm and running time analysis.
We also give a lower bound on the largest number of preferred extensions an oriented graph on $n$ vertices may have of $\Omega(3^{n/6}) \ge \Omega(1.2009^n)$.
A construction, which we call the \emph{Oriented Translation}, reducing an arbitrary digraph $G=(V,E)$ to an oriented graph with $|V|+r(G)$ vertices, such that there is a bijection between their preferred extensions, allows us to generalize these upper and lower bounds to $\ostar\left(\varphi^{2\cdot r(G)}\cdot\varphi^{n-r(G)}\right)\le O(1.5180^{r(G)}\cdot 1.2321^{n-r(G)})$ and $\Omega(3^{r(G)/3}\cdot 3^{(n-r(G))/6}) \ge \Omega(1.4422^{r(G)} \cdot 1.2009^{n-r(G)})$, respectively.

Our main technical contribution is the third algorithm. It relies on a parameterized enumeration algorithm and extensions of the recent monotone local search framework \cite{monotonelocalsearch}.
% For a subset of vertices $S$ of a graph $G=(V,E)$, a set $T\subseteq S$ is maximal admissible with respect to $S$ if it is admissible in $G$ and there is no $T'$ with $T\subsetneq T' \subseteq S$ that is also admissible in $G$.
The parameterized enumeration algorithm has as input a digraph $G=(V,E)$, a set of arguments $S$, and a non-negative integer $k$, and it enumerates all maximal admissible extensions $T \subseteq S$ of $G$ within distance $k$ of $S$.
% The parameterized enumeration algorithm has as input a digraph $G=(V,E)$, a set of arguments $S$, and a non-negative integer $k$, and it enumerates all admissible extensions $T\subseteq S$ of $G$ that are maximal with respect to $S$ and have distance at most $k$ from $S$ (i.e., they can be obtained from $S$ by removing at most $k$ vertices).
%Its running time can be upper bounded by $\ostar(2^{c\cdot k+(1-c)\cdot r(G[S])/2})$ for every $c$ with $1/2\le c \le 1$. This is optimal, since there are instances for which the solution consists of $\Omega(2^{c\cdot k+(1-c)\cdot r(G[S])/2})$ extensions for every $c$ with $1/2\le c \le 1$.
Its running time can be upper bounded by $\ostar(2^{k/2+r(G[S])/4})$.
This is optimal, since there are instances for which the solution consists of $\Omega(2^{k/2+r(G[S])/4})$ preferred extensions at distance at most $k$ from $S$.
Furthermore, under the Strong Exponential Time Hypothesis, the corresponding decision problem has no $\ostar(2^{(1-\varepsilon)(k/2+r(G[S])/4)})$ time solution for any $\varepsilon > 0$.
%We are not aware of any previous algorithms whose running time depends crucially on the number of vertices in $2$-cycles.
We use this parameterized enumeration algorithm in a new 2-phase monotone local search procedure, where we separately sample vertices from $B$, the set of vertices in at least one 2-cycle, and $V \setminus B$ and then apply the parameterized enumeration algorithm.
The running time analysis is a new combination of the results in \cite{monotonelocalsearch} for the first sampling phase and \cite{multivariatesubroutines} for the second sampling phase, combined with the parameterized subroutine.
From a technical point of view, this is the most innovative part of this paper. (From a conceptual point of view, the most innovative contribution is probably the synergy between modern enumeration algorithmics and the theory of abstract argumentation.)
This results in an algorithm enumerating all preferred extensions of a given digraph $G$ in time $\ostar\left( \left( 1+2^{1/4}-2^{-1/2}\right)^{(1+\varepsilon) \cdot r(G)} \cdot \left(2-2^{-1/2}\right)^{(1+\varepsilon) \cdot (n-r(G))} \right)$.

\smallskip
\noindent
\textbf{Interpretation of results}
\autoref{fig:run-time} illustrates the running time of the various algorithms. We have improved algorithms and combinatorial upper bounds whenever $r\le 0.8004$.
The result for $r=0$ shows that for oriented graphs, our new algorithm allows to handle instances with $75\%$ more arguments, compared with the previous best $O(3^{n/3})$ upper bound. (We have that $\log_{1.2321} (3^{1/3}) \approx 1.7545$.) The figure also shows that we have significantly narrowed the gap between the best known lower bound and the best known upper bound for oriented graphs, and digraphs with a fraction $r$ of vertices belonging to $2$-cycles, for a wide range of $r$.

\smallskip
\noindent
\textbf{Outline}
Sections \ref{sect:parameterizedintro} and \ref{monotonelocalsearch_sect} describe our monotone local search algorithm.
Section \ref{improvedenumeration_sect} describes our branching algorithm for oriented graphs.

First we introduce the parameterized enumeration problem that will form the subroutine of our Monotone Local Search.

\pbDefOptPara{Maximal Admissible Subset Enumeration (MASE)}{Graph $G$, set $S \subseteq V(G)$, integer $k$}{$k$}{Enumerate all maximal admissible sets $T \subseteq S$ such that $\abs{S \setminus T} \leq k$.}
There is a subtlety here with how we define maximal. We say $T$ is a maximal admissible subset of $S$ if there does not exist an admissible set $U$ such that $T \subsetneq U \subseteq S$.
Notably $T$ is not necessarily a preferred extension (though $T$ is if $S = V(G)$).

In \autoref{sect:parameterizedintro}, we present an algorithm for MASE, parameterized by $k$ and $r(G[S])$, the resolution order of the subgraph of $G$ induced by $S$.

\begin{theorem}
    \label{thm:OutlineMASEResult}
    For an instance $I = (G,S,k)$ of MASE, let $\mu(I) := \frac{k}{2} + \frac{r(G[S])}{4}$.
    Then MASE can be solved in $\ostar(2^{\mu(I)})$ time.
    Furthermore, there are at most $2^{\mu(I)}$ maximal admissible subsets of $S$ within distance $k$ of $S$. Hence, there are at most $2^{\mu(I)}$ preferred extensions that are subsets of $S$ with size $\geq \abs{S} - k$.
\end{theorem}
Our algorithm is a standard parameterized branching algorithm.
Compared to the enumeration of independent sets, the primary additional tool we have is a powerful simplification rule, \hyperref[simprule:undefendable]{\textbf{(Undefendable)}}, for vertices with in-degree $0$.
\hyperref[simprule:undefendable]{\textbf{(Undefendable)}} also allows our base case to be any conflict-free set. % TODO: wording?

In \autoref{monotonelocalsearch_sect}, we extend our parameterized algorithm into a general enumeration algorithm through a novel 2-phase application of monotone local search.
Since 2-cycles increase the run time of our MASE subroutine, we modify the classical Monotone Local Search to sample separately between a set of "bad vertices" (ones contained in a 2-cycle) and "good vertices" (ones not contained in any 2-cycle).
This presents a speed up compared to a more direct application of the Monotone Local Search framework.
We believe this may be useful for other problems.

Separately, \autoref{improvedenumeration_sect} presents a branching algorithm for oriented graphs. Again, \hyperref[simprule:undefendable]{\textbf{(Undefendable)}} plays a critical role.
This time around, we allow our base case to be any induced DAG which further provides a simplification rule for vertices with out-degree 0.
We tailor our branching rules to take full advantage of these two simplification rules.
This is combined with a lot of careful case analysis and ad-hoc methods (including a graph classification theorem in \autoref{subsect:Oriented_Case4}).

Each of our algorithms also provide a corresponding combinatorial upper bound on the number of preferred extensions.
The various enumeration algorithms and bounds are collected in \autoref{bounds_sect}.

In the appendix we have a mix of background, further detail and extra results. Some extra results that may be of independent interest:

\appendixref{sect:orientedtranslation} describes the \emph{Oriented Translation}, a parsimonious reduction from a general directed graph to an oriented graph.

\begin{theorem}
    \label{outline:orientedtranslation}
%Let $G$ be an AF. % Let $r(G)$ be the number of vertices in $G$ that are in at least one 2-cycle.
There is a linear time algorithm that transforms any AF $G$ into an oriented AF $G'$ with $\abs{V(G')} = \abs{V(G)} + r(G) + 3$
such that
there is a
bijection between the preferred extensions of $G$ and the preferred extensions of $G'$ that can be applied in linear time.
\end{theorem}

The basic idea of the construction (see \appendixref{sect:orientedtranslation}) is carefully converting 2-cycles into 4-cycles by doubling up the vertices that are contained in at least one 2-cycle.

Our primary interest in \autoref{outline:orientedtranslation} is as a tool for extending oriented graph algorithms to algorithms on general graphs, parameterized by resolution order. In \autoref{bounds_sect} we apply this to our branching algorithm for oriented graphs.

\autoref{outline:orientedtranslation} is also useful for deriving complexity results on oriented graphs by extending constructions for directed graphs. In \appendixref{sect:nooutputpoly} we obtain the following result:

\begin{theorem}
Unless P=NP, no algorithm enumerates the admissible or preferred extensions of an AF in output-polynomial time, even when the AF is an oriented graph.
\end{theorem}

In \appendixref{sect:complexitylowerbounds} we also show our algorithm for the decision variant of MASE is optimal, assuming the Strong Exponential Time Hypothesis,
and that the enumeration bound in \autoref{thm:OutlineMASEResult} is tight (i.e: there exist instances with $2^{\mu(I)}$ maximal admissible subsets within distance $k$).

We briefly justify the choice of the measure $\mu$ we use for our MASE algorithm in \appendixref{sect:discussionofmeasure}.

% TODO: reset numbering?
\smallskip
\noindent
\textbf{Notation}
An \emph{oriented graph} is a digraph with no 2-cycles.

%
%A digraph $G = (V,E)$ is oriented if $G$ has no 2-cycle, or equivalently, there does not exist $a, b \in V$ such that $\edge{a,b}, \edge{b,a} \in E$.
%We will also sometimes say $a$ attacks $b$ to mean $\edge{a,b} \in E$. Similarly, we will say $b$ is attacked by $a$ if $\edge{a,b} \in E$.
We will use the notation $N(v)$ to denote the set of vertices adjacent to $v$ and $N[v]$ to denote $N(v) \cup \{v\}$.

%
%We define:
%\begin{definition}[Resolution Order]
%\label{definition:resolutionorder}
%The \emph{resolution order} of a digraph $G=(V,E)$ as the number of vertices that belong to a 2-cycle in $G$.
%\end{definition}
%
We will assume throughout that AFs have no self loops. We can safely make this assumption due to the Loopless Translation (\autoref{looplesstranslation}) which, with a (small) constant overhead removes all self loops from an AF while preserving the admissible extensions.

\section{Parameterized Enumeration Problems}
\label{sect:parameterizedintro}
% TODO: if we have space, put an overview here too.
Our algorithm for MASE will follow a standard template for parameterized branching algorithms. A summary of the necessary concepts can be found in \appendixref{sect:background-branching}.

We will consider the measure $\mu(I) = \frac{k}{2} + \frac{b}{4}$ where $k$ is the number of vertices we are allowed to remove and $b := r(G[S])$ is the resolution order of $G[S]$.
We will design a branching algorithm with run time $\ostar(2^{\mu(I)})$ which recurses into subinstances and collates their results to obtain the maximal admissible subsets of $S$.

However, there is a technical difficulty in the collation step arising from the fact that a maximal admissible subset of $S' \subsetneq S$ may not be a maximal admissible subset of $S$. This gives rise to the following subproblem:

\pbDefOpt{Maximal Subset Collation}{Graph $G$, $c$ pairs $(S_i, C_i)$ where for each $i$, $S_i \subseteq V(G)$ and $C_i$ is a set containing only maximal admissible subsets of $S_i$}{A set containing all maximal elements of $\bigcup_{i=1}^{c} C_i$}

The main result we need is:
\begin{lemma}
    \label{lemma:collationruntime}
    % TODO: check the ostar?
    Maximal Subset Collation can be solved in $O(c\sum_{i=1}^{c} \abs{C_i}\cdot\poly(\abs{V}))$ time. 
\end{lemma}

An algorithm for Maximal Subset Collation is provided in \appendixref{sect:maximalsubsetcollation}.

\vspace{2mm}
We are now ready to solve MASE.

% TODO: Maybe make a note on how this is the standard technique we use? Look at degree measures.
% TODO: no longer a subsubsection?
\subsection{\texorpdfstring{An $\ostar(2^{\mu(I)})$}{An } algorithm for Maximal Admissible Subset Enumeration}
% TODO: Would it be cleaner to add more state to the parameterized problem?
The overall structure of our algorithm is described in \autoref{alg:MASEStructure}.

The following table lists the branching rules in application order with the first row being the base case. The second column describes the instances each rule is applicable to. After applying simplification rules, our branching algorithm will always apply the first rule that is applicable to the input instance.

\begin{center}\small
	\newcolumntype{L}[1]{>{\hsize=#1\hsize\raggedright\arraybackslash}X}

    \begin{tabularx}{1\textwidth}{L{0.15} L{1.3} L{0.8}}
\toprule
Case & Requirement to apply & Running Time \\
\midrule%\hline
Base & $S$ is conflict-free. & Solves in $O^{*}(1)$ time, returns $\leq 1$ set. \\
%\hline
1 & $G[S]$ is oriented with maximum total degree $\leq 2$. & Branching vector $(1,1)$, branching number $2$. \\
%\hline
2 & There is a 2-cycle in $G[S]$. & Branching vector $(1,1)$, branching number $2$. \\
%\hline
3 & $G[S]$ has maximum total degree $\ge 4$. & Branching vector $(2,\frac{1}{2})$, branching number $\approx 1.91$. \\
%\hline
4 & $G[S]$ contains a vertex with total degree 3. & Branching vector $(1,\frac{3}{2})$, branching number $\approx 1.76$. \\
\bottomrule
\end{tabularx}
\end{center}

\noindent
We note that these requirements are exhaustive, hence there will always be at least one applicable rule in any instance.

\begin{algorithm}[tb]
\caption{Structure of MASE branching algorithm}
\label{alg:MASEStructure}
\begin{algorithmic}
\Ensure{Returns all maximal admissible subsets $T \subseteq S$ such that $\abs{S \setminus T} \leq k$} 
\Function{MASE}{$S, k$}
	\If{$k < 0$}
		\State \Return $\emptyset$
	\EndIf
	\While{\textbf{(Undefendable)} applies}
		\State Apply \textbf{(Undefendable)}
	\EndWhile
	\If{Base Case applies}
		\State Solve the instance directly through the base case subroutine.
	\Else
		\State Let $b_i$ be the first branching rule that applies.
		\State Let $(S_1, k_1) \ldots (S_r, k_r)$ be the subinstances obtained from applying $b_i$.
		\State Let $C_i = \mathrm{MASE}(S_i, k_i)$, for all $1 \leq i \leq r$.
		\State \Return $\mathrm{Maximal\_Subset\_Collation}((S_1, C_1), \ldots, (S_r, C_r))$
	\EndIf
\EndFunction
\end{algorithmic}
\end{algorithm}

We include the base case and short remarks on each of the branching cases. Full proofs of each case can be found in \appendixref{sect:MASE-cases}.

The key ingredient for solving the base case and optimizing our branching cases is the following simplification rule.
It is essentially a rephrasing of a well known result that has been applied before in the context of enumeration algorithms (e.g., \cite{preferredextensionsquery}\cite{semistablealgo}).

\begin{simprule}[\textbf{Undefendable}]
    \label{simprule:undefendable}
    Suppose we are trying to enumerate all maximal admissible subextensions of $S$.

    Let $u \in S$ be a vertex such that there exists a vertex $a \in V(G)$, $a$ attacks $u$ and no vertex in $S$ attacks $a$. Then there is no admissible subset of $S$ that contains $u$ so we can safely set $S \leftarrow S \setminus \{u\}$.
\end{simprule}

\begin{lemma}
    \textbf{(Undefendable)} is sound.
\end{lemma}
\begin{proof}
    Vertex $u$ can never be defended from $a$ in any subextension of $S$ and hence is in no admissible subextension of $S$.
\end{proof}

Instead of solving the base case directly we will instead solve the more general case where $G[S]$ is a DAG.
The increased generality will be needed in \autoref{improvedenumeration_sect}.

The base case, that $S$ is conflict-free, follows trivially (since $G[S]$ is a DAG if $S$ is conflict-free).

\subsection{MASE when \texorpdfstring{$G[S]$}{G[S]} is a DAG}
The results in this section are all well-known and follow directly from \hyperref[simprule:undefendable]{\textbf{(Undefendable)}}.

\begin{lemma}
    \label{lem:Undefendablenotapplicable}
    If $G[S]$ is a DAG, the rule \textbf{(Undefendable)} is not applicable if and only if $S$ is admissible.
\end{lemma}

\begin{proof}
    Clearly \textbf{(Undefendable)} is not applicable if and only if every $v \in S$ is acceptable with respect to $S$ (recall this means that for all arcs $(u,v) \in E(G)$, there is an arc $(w,u) \in E(G)$ with $w \in S$).
    By the definition of admissibility, it remains to show that if \textbf{(Undefendable)} is not applicable, then $S$ is conflict-free.
    This follows from noting \textbf{(Undefendable)} is applicable to any vertex attacked by a maximal vertex in $G[S]$.
    Such a vertex exists in any weakly connected DAG with more than one vertex.
    Hence all weakly connected components of $G[S]$ have size one.
\end{proof}

\noindent
\textbf{(Undefendable)} takes polynomial time to apply. The resulting subset from applying \textbf{(Undefendable)} until it is not applicable is the unique maximal admissible subset of $S$. It is admissible by \autoref{lem:Undefendablenotapplicable}. It is the unique maximal admissible subset as \textbf{(Undefendable)} only removes vertices that are in no admissible subsets of $S$. Hence:

\begin{lemma}
    \label{lemma:DAGmaximaladmissible}
	Suppose $S \subseteq V(G)$ induces a DAG in $G$. Then $S$ has exactly one maximal admissible subset and it can be found in polynomial time.
\end{lemma}

\subsection{Remarks on branching cases}
Full proofs can be found in \appendixref{sect:MASE-cases}.

\hyperref[subsect:MASE_Case1]{Case 1} follows from noting $G[S]$ must be a family of cycles.

Let $v$ be the vertex in the 2-cycle with higher total degree. 
\hyperref[subsect:MASE_Case2]{Case 2} follows from a 2-way branch, in one branch enumerating all maximal admissible subsets that include $v$, in the other branch enumerating all that do not include $v$.

\hyperref[subsect:MASE_Case3]{Case 3} is another 2-way branch, except applied to any vertex with total degree at least 4.

\hyperref[subsect:MASE_Case4]{Case 4} is probably the most instructive.
It best showcases the power of \hyperref[simprule:undefendable]{\textbf{(Undefendable)}}. 
First we show a vertex $v$ exists with in-degree 1, out-degree 2.
Then a 2-way branch is applied to the vertex attacking $v$, using \hyperref[simprule:undefendable]{\textbf{(Undefendable)}} to improve the branch where the vertex is excluded and $v$ is left with in-degree 0.

% We defer the details on handling each case to \appendixref{sect:MASE-cases}. We include the base case and case 4, the latter as an example of the general techniques used in the case analyses. As a summary:
% We include the base case and defer the details on handling each case to \appendixref{sect:MASE-cases}.
% \paragraph*{Base Case - $S$ is conflict-free}
% By \autoref{lemma:DAGmaximaladmissible}, $S$ has a unique maximal admissible subset, $T$, which can be found in polynomial time. Return this subset if $\abs{S \setminus T} \leq k$.

% TODO: maybe remove this?
% \paragraph*{Case 4 - A vertex with total degree 3 exists}
% As there is a degree 3 vertex, not all vertices have the same in-degree as out-degree. Since the sum of in-degrees in $G[S]$ equals the sum of out-degrees, there exists a vertex with higher out-degree than in-degree.
% 
% As no vertices have in-degree 0 (due to our simplification rule), $G[S]$ necessarily has a vertex, say $v$, with in-degree 1 and out-degree 2. We branch on the in-neighbor, say $u$ of such a vertex. If we include $u$, we note $u$ has at least 2 neighbors (as $G$ is oriented and in-degrees can not be zero), hence $\mu(I)$ decreases by at least 1.
% 
% Otherwise, we exclude $u$, allowing us to apply \textbf{(Undefendable)} to $v$. In this case, $\mu(I)$ decreases by at least $\frac{3}{2}$. Hence we obtain a branching rule with branching vector $(1,\frac{3}{2})$ and branching number $\approx 1.76 < 2$.
% 
\subsection{Running Time Analysis}
\label{subsect:MASE_Runningtime_Analysis}
The base case is solvable in polynomial time. Each of our branching rules has branching number $\leq 2$.
Hence, if we ignore the calls to Maximal Subset Collation, by the Combine Analysis Lemma (\autoref{combineanalysislemma}) our algorithm has running time $\ostar(2^{\mu(I)})$.

% TODO: this part is pretty poorly explained, we account for the cost of an instance in its parent, not in the instance itself.
Maximal Subset Collation is applied to each admissible subset encountered by the MASE algorithm (i.e: each leaf in the search tree) at most $d$ times, where $d$ is the maximum depth of the search tree.
By \autoref{lemma:collationruntime} each application incurs a $O(\poly(\abs{V}))$ cost ($c\leq 2$ for our branching rules). Hence the overall cost incurred by the Maximal Subset Collation step is
\begin{center}
    $O(A \cdot d \cdot \poly(\abs{V}))$
\end{center}
where $A$ is the total number of admissible subsets encountered by the MASE algorithm (equivalently, the number of leaves in the search tree).
By the Combine Analysis Lemma for Enumeration (\autoref{combineanalysislemmaforenumeration}), $A$ is $O(2^{\mu(I)})$.
The maximum depth $d$ is $O(\mu(I))$ which we may take to be $O(\abs{V})$.
Hence Maximal Subset Collation incurs an overhead cost of $\ostar(2^{\mu(I)})$.

% We note that the Maximal Subset Collation subroutine (\autoref{alg:collation}) does not incur overhead as the search tree has depth at most $k$, hence Maximal Subset Collation is applied to each admissible subset at most $k$ times.

As noted, our algorithm also satisfies the requirements for applying the Combine Analysis Lemma for Enumeration (\autoref{combineanalysislemmaforenumeration}). We summarize all the above results in the following theorem.

\begin{theorem}
    \label{parameterized_adm_rem}
    Let $\mu(I) = \frac{k}{2} + \frac{b}{4}$, where $b$ is the resolution order of $G[S]$.
	Then MASE can be solved in $\ostar(2^{\mu(I)})$ time.
    Furthermore, there are at most $2^{\mu(I)}$ maximal admissible subsets of $S$ within distance $k$ of $S$. Hence, there are at most $2^{\mu(I)}$ preferred extensions that are subsets of $S$ with size $\geq \abs{S} - k$.
\end{theorem}

\section{Monotone Local Search}
\label{monotonelocalsearch_sect}

In this section, we apply Monotone Local Search to our $\ostar(2^{\mu(I)})$ algorithm for Maximal Admissible Subset Enumeration. A basic exposition of Monotone Local Search will be provided here, additional background can be found in \appendixref{sect:background-mls}.

% TODO: PROBS NEED TO SAY THAT 
The framework normally applies to extension problems. However we can just as easily apply it to removal problems (formally by focusing on the complement of each set, we can turn any removal problem into an extension problem). Hence, we will freely use the framework with removal problems instead.

%We also note that the set of preferred extensions of an AF is not polynomial-time computable, unless P=NP. In fact, it is NP-hard to determine whether even the empty set is a preferred extension \cite{graphtheoreticalstructures}.
%Therefore, we require that a set of admissible extensions is enumerated that includes all the preferred extensions, but may also include non-maximal admissible extensions.

%To apply the framework, we set Maximal Admissible Removal as our $\Phi$-Removal problem (Maximal Admissible Removal is just the decision problem version of Maximal Admissible Subset Enumeration).
For our application, the instance $I$ is the graph $G$ and the family we are looking to enumerate, $\cF_I$, is the set of all preferred extensions of $G$. We will apply Monotone Local Search using MASE as our subroutine.

For a MASE instance $I'=(G,X,k)$, let $\cF_{I,X}^k$ denote the set of all maximal admissible sets $T \subseteq X$ such that $\abs{X \setminus T} \leq k$.
Hence $\cF_{I,X}^k$ contains all preferred extensions that are subsets of $X$ within distance $k$, but may also contain admissible extensions that are not preferred extensions.

Because of this, our Monotone Local Search will enumerate $F_I$, however, it may also enumerate some non-maximal admissible extensions. Our result, \autoref{theorem:MLSresult}, will account for this, however, for simplicity we will just speak of enumerating $F_I$ throughout this section.
% TODO: maybe write more on this subtlety here.

A naive application of the Monotone Local Search framework with our $\ostar\left(2^{\frac{k}{2}+\frac{b}{4}}\right)$ algorithm for enumerating $\cF_{I,X}^k$  yields an $\ostar\left(2^{\frac{b}{4}}\left(2-\frac{1}{\sqrt{2}}\right)^{n+o(n)}\right) \approx \ostar(1.1893^b \cdot 1.2929^n)$ time algorithm that enumerates all preferred extensions of $G$. To improve this we need a basic understanding of how Monotone Local Search works.

\subsection{Basic Overview of Monotone Local Search}
We need the following definition from \cite{monotonelocalsearch} (slightly modified to account for our preference for removal problems):
\begin{definition}[\cite{monotonelocalsearch}]
	Let $U$ be a universe of size $n$ and let $0 \leq p \leq q \leq n$. A family $C \subseteq \binom{U}{q}$ is an $(n,p,q)$-set-containing-family if for every set $S \in \binom{U}{p}$, there exists a $Y \in C$ such that $S \subseteq Y$.
\end{definition}

For any fixed $s$, a value $t \geq s$ will somehow be chosen. Then, a $(n,s,t)$-set-containing-family $C_s$ is constructed. For each set in $C_s$, its subsets from $\cF_I$ obtained by removing at most $t-s$ elements are then enumerated using an enumeration subroutine. This enumerates all elements of $\cF_I$ with size $s$. Supposing the subroutine has run time $\ostar(\alpha^k)$ (where $k$ is the parameter), this step of Monotone Local Search has run time $\ostar(\abs{C_s}\cdot\alpha^{t-s})$.

Repeating this for all $s$, Monotone Local Search has running time $\ostar(\max\limits_{1 \leq s \leq n} \abs{C_s}\cdot\alpha^{t-s})$. With the right choices of $t$, \cite{monotonelocalsearch} shows the running time $\ostar((2-\frac{1}{\alpha})^{n+o(n)})$.

\subsection{Improving our Monotone Local Search}
% TODO: don't want to define B and D here cause the algorithm might be before it.
We start with some notation. For any digraph $G = (V,E)$, let $B$ be the set of vertices in $V$ in at least one 2-cycle and let $D := V \setminus B$. The key idea is we will sample vertices from $B$ and $D$ separately. The overall structure is described in \autoref{alg:improvedMLS}. Except for the separate sampling, it is identical to a standard application of Monotone Local Search.

First, we argue correctness, i.e: that every preferred extension is enumerated at least once. Fix a preferred extension, say $U$, and suppose $U$ contains $b$ vertices in $B$ and $d$ vertices in $D$.
Then, by the definition of set-containing-families, there exists a $S \in C_b$ such that $U \cap B \subseteq S$ and a $T \in C_d$ such that $U \setminus B \subseteq T$. Now we note that $S \subseteq B, T \subseteq V \setminus B$ to get:
\[
    \abs{(S \sqcup T) \setminus U} = \abs{S \setminus (U \cap B)} + \abs{T \setminus (U \setminus B)} = (b'-b) + (d'-d)
\]
Hence $U$ is enumerated in the call to $\mathrm{MASE}(S \cup T, (b' - b) + (d' - d))$ as required.

Now we argue the runtime. For a fixed $b, d$ the calls to MASE have total run time $\ostar(\abs{C_b}\cdot\abs{C_d}\cdot2^{\frac{(b'-b)+(d'-d)}{2}+\frac{b'}{4}})$ for some choice of $b', d'$.
Our overall running time is
\begin{center}
    $\ostar\left(\max\limits_{1\leq b \leq \abs{B}}\max\limits_{1 \leq d \leq \abs{D}} \abs{C_b}\cdot\abs{C_d}\cdot2^{\frac{(b'-b)+(d'-d)}{2}+\frac{b'}{4}}\right)$
\end{center}
We can split this into two terms to get a complexity of $\ostar\left(\max\limits_{1 \leq d \leq \abs{D}} \abs{C_d}\cdot2^{\frac{(d'-d)}{2}}\right)$
multiplied by $\ostar\left(\max\limits_{1\leq b \leq \abs{B}} \abs{C_b}\cdot2^{\frac{(b'-b)}{2}+\frac{b'}{4}}\right)$.

\begin{algorithm}[tb]
\caption{Structure of Improved Monotone Local Search algorithm}
\label{alg:improvedMLS}
\begin{algorithmic}
    \Ensure{Returns a set containing all preferred extensions of $G$ (and possibly some non-preferred extensions)}
\Function{ImprovedMLS}{$G$}
    \State Let $\cF_I = \emptyset$.
    \For{$b = 1$ to $\abs{B}$}
        \State Let $b' = \mathrm{determine\_b'}(b)$.
        \State Let $C_b$ be a $(\abs{B},b,b')$-set-containing-family.
        \For{$d = 1$ to $\abs{D}$}
            \State Let $d' = \mathrm{determine\_d'}(d)$.
            \State Let $C_d$ be a $(\abs{D},d,d')$-set-containing-family.
            \ForAll{pairs $(S,T)$ with $S \in C_b, T \in C_d$}
                \State Let $\cF_I = \cF_I \cup \mathrm{MASE}(S \cup T, (b' - b) + (d' - d))$.
            \EndFor
        \EndFor
    \EndFor
    \Return $\cF_I$
\EndFunction
\end{algorithmic}
\end{algorithm}

% Now, we fix a pair $b \leq \abs{B}, d \leq \abs{D}$ and enumerate all subsets of $\cF_I$ that contain $b$ vertices in $B$ and $d$ vertices in $D$. As before, we will somehow choose values $b' \geq b$ and $d' \geq d$ and construct a $(\abs{B},b,b')$-set-containing-family $C_b$ and a $(\abs{D},d,d')$-set-containing-family $C_d$. Then for each pair $S \in C_b, T \in C_d$, we run our parameterized MASE algorithm from \autoref{parameterized_adm_rem} on $S \cup T$ with parameter $k = (b'-b) + (d'-d)$. This step of Monotone Local Search has run time $\ostar(\abs{C_b}\cdot\abs{C_d}\cdot2^{\frac{(b'-b)+(d'-d)}{2}+\frac{b'}{4}})$.

We will now analyze the running time with the right choices of $b'$ and $d'$. \autoref{thm:appendixmls} summarizes the results we use to analyze the complexity.

% TODO: ITS CONFUSING HOW I SWITCH INTO COMPLEXITY ANALYSIS JUST DIRECTLY USING THE RESULTS WITHOUT AN ALGORITHM TO ANALYSE.
The first of these two terms can be analyzed using the analysis in \cite{monotonelocalsearch}.
This term is the complexity one attains for Monotone Local Search with a $\ostar(2^{\frac{k}{2}})$ subroutine. Hence the analysis in \cite{monotonelocalsearch} gives a complexity of $\ostar((2-\frac{1}{\sqrt{2}})^{\abs{D}+o(\abs{D})})$.

We need the extended analysis presented in \cite{multivariatesubroutines} to analyze the second term.
This term is the complexity one attains for Monotone Local Search with a $\ostar(2^{\frac{k}{2}+\frac{n-\abs{X}}{4}})$ subroutine (where $k$ is the parameter, $n = \abs{U}$ the size of the underlying set and $X$ is the set we are extending). Hence the analysis in \cite{multivariatesubroutines} gives a complexity of $\ostar((1+2^{\frac{1}{4}} - \frac{1}{\sqrt{2}})^{\abs{B}+o(\abs{B})})$.

\medskip
The papers we have cited also give a corresponding combinatorial upper bound on the number of preferred extensions. We summarize the above results as follows.

\begin{theorem}
    \label{theorem:MLSresult}
	Let $G = (V,E)$ be a digraph. Let $r$ be the proportion of vertices in $V$ that are in at least one 2-cycle.
	
	Then there exists a $\ostar(((1+2^{\frac{1}{4}}-\frac{1}{\sqrt{2}})^r(2-\frac{1}{\sqrt{2}})^{1-r})^{\abs{V}+o(\abs{V})})$ $\approx$  $\ostar((1.4822^r1.2929^{1-r})^{\abs{V}})$ time algorithm that enumerates all preferred extensions of $G$; however it may also enumerate some non-maximal admissible extensions.	
	Furthermore, there are at most $\ostar(((1+2^{\frac{1}{4}}-\frac{1}{\sqrt{2}})^r(2-\frac{1}{\sqrt{2}})^{1-r})^{\abs{V}})$ $\approx$  $\ostar((1.4822^r1.2929^{1-r})^{\abs{V}})$ preferred extensions in $G$.
\end{theorem}

\section{Improved Enumeration Algorithm for Oriented Graphs}
\label{improvedenumeration_sect}
Finally, we outline a branching algorithm with a finer analysis for oriented graphs.
As with MASE, we will follow a standard template for branching algorithms.
A summary of the necessary concepts can be found in \appendixref{sect:background-branching}.
Our algorithm creates a search tree with at most $\varphi^n$ leaves where $\varphi \approx 1.2321$ is the branching number for branching vector $(8,1)$.

In \autoref{bounds_sect} we will use the Oriented Translation (\autoref{outline:orientedtranslation}) to obtain a general enumeration algorithm parameterized by the number of vertices in at least one 2-cycle.

\subsection{Overview}
The overall structure of our algorithm is described in \autoref{alg:OrientedAlgoStructure}.

The state of the algorithm consists of the subset $\Und \subseteq V(G)$ and a queue of vertices $\Def$ ($\Und$ for undecided and $\Def$ for deferred). $\Und$ is the set of vertices we have yet to make a decision on whether we should include them. $\Def$ is a queue of vertices which have no outgoing arcs to $\Und$. These vertices will be handled in the base case. While branching we can essentially assume the vertices in $\Def$ do not exist.

We maintain the following invariants. We outline why they are invariant, it is straight-forward to verify that each case of our branching algorithm maintains these invariants.
\begin{enumerate}
	\item $\Und$ and $\Def$ are disjoint. This holds as vertices are only ever deleted from or moved between $\Und$ and $\Def$, never copied. \label{invariant:1}
	\item $G[\Def]$ is a DAG where each vertex $v \in \Def$ only attacks vertices that were added to $\Def$ before $v$. This is crucial for the base case since a DAG has 1 maximal admissible subset. This holds as only vertices with out-degree $0$ in $G[\Und]$ are ever moved to $\Def$. \label{invariant:2} % TODO: Cite the exact simplification rule.
	\item There is no attack from any vertex in $\Def$ to any vertex in $\Und$. This holds for the same reason as Invariant 2.
\end{enumerate}

\begin{algorithm}[tb]
	\caption{Structure of Oriented maximal admissible enumeration algorithm}
	\label{alg:OrientedAlgoStructure}
	\begin{algorithmic}
		\Require $\Und \cap \Def = \emptyset$.
		\Require $G[\Def]$ is a DAG where each $v \in \Def$ only attacks vertices added to $\Def$ before $v$.
		\Require There is no attack from any vertex in $\Def$ to any vertex in $\Und$.
		\Ensure{Returns all maximal admissible subsets of $\Und \cup \Def$.}
		\Function{OrientedEnumeration}{$\Und, \Def$}
		\If{$k < 0$}
		\State \Return $\emptyset$
		\EndIf
		\While{Any simplification rule applies}
		\State Apply said simplification rule
		\EndWhile
		\If{Base Case applies}
		\State Solve the instance directly through the base case subroutine.
		\Else
		\State Let $b_i$ be the first branching rule that applies.
		\State Let $(U_1, D_1) \ldots (U_r, D_r)$ be the subinstances obtained from applying $b_i$.
		\State Let $C_i = \mathrm{OrientedEnumeration}(U_i, D_i)$, for all $1 \leq i \leq r$.
		\State \Return $\mathrm{Maximal\_Subset\_Collation}
		((U_1 \cup D_1, C_1), \ldots, (U_r \cup D_r, C_r))$
		\EndIf
		\EndFunction
	\end{algorithmic}
\end{algorithm}

Our measure is $\mu = \abs{\Und}$. For any instance, our algorithm creates a search tree with at most $\varphi^\mu$ leaves. Hence, calling $\mathrm{OrientedEnumeration}(V(G), [])$ will return all preferred extensions of $G$ by traversing a search tree with at most $\varphi^{|V(G)|}$ leaves.

\subsection{Extra Notation}
\label{subsect:Oriented_ExtraNotation}
We will say a vertex has degree $(a,-)$ if it has in-degree $a$, a vertex has degree $(-,b)$ if it has out-degree $b$ and a vertex has degree $(a,b)$ if it has in-degree $a$ and out-degree $b$.

\subsection{Simplification Rules}
Both of these are applicable in polynomial time and decrease $\mu = \abs{\Und}$.

\begin{simprule}[Out-degree 0]
\label{simprule:outdeg0}
Let $v$ be a vertex in $G[\Und]$ with out-degree $0$. Move $v$ from $\Und$ to the end of the queue $\Def$.
\end{simprule}

\begin{simprule}[In-degree 0]
\label{simprule:indeg0}
Let $v$ be any vertex in $G[\Und]$ with in-degree $0$. Then by Invariant 3, $v$ has in-degree $0$ in $G[\Und \cup \Def]$.
Applying Simplification Rule \textbf{(Undefendable)}
we can set $\Und \gets \Und \setminus N(v)$. After that, $v$ has out-degree 0 in $G[\Und]$ and hence we move $v$ from $\Und$ to $\Def$.

Our new instance $I' = (\Und', \Def')$ has:
\begin{itemize}
    \item $\Und' = \Und \setminus N[v]$.
    \item $\Def' = \Def \cup \{v\}$.
\end{itemize}
\end{simprule}

Due to these rules, henceforth we may assume each vertex in $G[\Und]$ has in-degree $\geq 1$, out-degree $\geq 1$ and (total) degree $\geq 2$.

% TODO: Remove if we can compress this down.
% \newpage
\subsection{Branching Rules}
Our algorithm will apply the first rule applicable to the instance:

\begin{center}
	\newcolumntype{L}[1]{>{\hsize=#1\hsize\raggedright\arraybackslash}X}%
    \begin{tabularx}{1\textwidth}{L{0.15} L{1.3}  L{0.8} }
		%\hline
		\toprule
        Case & Requirement to apply & Worst case branching number \\
		\midrule%\hline
        Base & $\Und = \emptyset$. & Solves in $O^{*}(1)$, returns 1 set. \\
		%\hline
        1 & $\exists v \in G[\Und]$ with total degree $\geq 7$. & Branching vector $(8, 1)$, branching number $\varphi \approx 1.2321$. \\
		%\hline
        2 & $\exists v \in G[\Und]$ with degree $(1,-)$. & Branching vector $(4, 3)$, branching number $\approx 1.221$. \\
		%\hline
        3 & $\exists v \in G[\Und]$ with in-degree $\neq$ out-degree. & Branching vector $(6, 5, 5)$, branching number $\approx 1.2298$. \\
		%\hline
        4 & $G[\Und]$ has a weakly connected component where every vertex has degree $(2,2)$. & Branching vector $(6, 5, 5)$, branching number $\approx 1.2298$. \\
		%\hline
        5 & $G[\Und]$ has a weakly connected component where every vertex has degree $(3,3)$. & Branching vector $(7, 7, 7, 7)$, branching number $\approx 1.219$. \\
		%\hline
        6 & There is a weakly connected component in $G[\Und]$ where every vertex has in-degree $=$ out-degree. & Branching vector $(7, 5, 5)$, branching number $\approx 1.218$. \\
		\bottomrule
	\end{tabularx}
\end{center}

We note the base case and cases 3 and 6 cover all possible inputs. Hence there will always be at least one applicable rule.

Our primary strategy is the following branching rule applied to a specifically chosen $v$.
\subsection{2-way branch on whether to include \texorpdfstring{$v$}{v}}
\label{subsect:Oriented_2waybranch}
We will often pick a specific vertex $v$ and do a 2-way branch, in one branch enumerating all maximal admissible subsets that include $v$, in the other branch enumerating all that do not include $v$.

\begin{itemize}
    \item In the branch where we decide to include $v$ we can not include any of $v$'s neighbors. Hence we create a new instance $I' = (\Und', \Def')$ where $\Und' = \Und \setminus N(v)$.

        Now $v$ is isolated in $G[\Und']$ so by \autoref{simprule:appendix_outdeg0} we may move $v$ from $\Und'$ to $\Def'$. Hence in our new instance we finally have:
    \begin{itemize}
        \item $\Und' = \Und \setminus N[v]$. % TODO: formatting here is a bit odd.
        \item $\Def' = \Def \cup \{v\}$.
    \end{itemize}
    and hence $\mu(I') = \mu(I) - \abs{N[v]} = \mu(I) - \deg(v) - 1$.

    As a technical note, not every subset enumerated in this branch contains $v$, however, every maximal admissible subset that contains $v$ will be enumerated in this branch.
    \item In the other branch we decide to exclude $v$. Hence our new instance $I' = (\Und', \Def')$ has:
    \begin{itemize}
        \item $\Und' = \Und \setminus \{v\}$.
        \item $\Def' = \Def$.
    \end{itemize}
    and hence $\mu(I') = \mu(I) - 1$.
\end{itemize}
Hence branching on whether we should include $v$ leads to a branching rule with branching vector $(\deg(v) + 1, 1)$.

\begin{definition} [\emph{2-way branch on whether to include $v$}]
    \label{def:2waybranch_includev}
    The phrase \emph{2-way branch on whether to include $v$} will be used as shorthand for the above branching rule.
\end{definition}

\subsection{Remarks on cases}
Full proofs can be found in \appendixref{sect:detailedcaseanalysis}.

In \hyperref[subsect:Oriented_BaseCase]{the base case}, by Invariant \ref{invariant:2}, $G[\Und \cup \Def]$ is a DAG. The base case then follows trivially from \autoref{lemma:DAGmaximaladmissible}.

\hyperref[subsect:Oriented_Case1]{Case 1} follows from a \branchon{v}, where $v$ is any vertex with total degree at least 7.

\hyperref[subsect:Oriented_Case2]{Case 2} is a good example of the strategy of picking a specific $v$ to apply a 2-way branch on.
We will pick a $v$ that allows us to apply our simplification rules in the branch where $v$ is excluded.
However, the choice of $v$ will depend on some case analysis.

\hyperref[subsect:Oriented_Case3]{Case 3} is similar to Case 2, a specific $v$ is chosen on which we apply a 2-way branch.

\hyperref[subsect:Oriented_Case4]{Case 4} is done through a graph classification theorem(\autoref{appendix:struct2regular}).

\hyperref[subsect:Oriented_Case5]{Case 5} is just a 4-way branch on any vertex $v$ and its $3$ in-neighbors.

For \hyperref[subsect:Oriented_Case6]{Case 6}, we first show there is a vertex with degree $(3, 3)$ attacking a vertex with degree $(2, 2)$, say $b$.
Then we apply a 3-way branch on $b$ and its $2$ in-neighbors.

\subsection{Summary of results}
In our base case we enumerate 1 extension in polynomial time.

Otherwise, we apply a branching rule with branching number $\leq \varphi$.

As in MASE, the Maximal Subset Collation subroutine (\autoref{alg:collation}) does not incur additional overhead as the search tree's depth is bounded by $\abs{V(G)}$ (see \autoref{subsect:MASE_Runningtime_Analysis} for the argument which we can apply verbatim).

Applying the Combine Analysis Lemma(\autoref{combineanalysislemma}) and Combine Analysis Lemma for Enumeration (\autoref{combineanalysislemmaforenumeration}) we obtain:

\begin{theorem}
	\label{betterorientedbound}
	Let $G = (V,E)$ be an oriented graph. Then there is an algorithm that enumerates all preferred extensions of $G$ with running time $\ostar(\varphi^{\abs{V}})$ where $\varphi$ is the unique positive root of $1 - x^{-1} - x^{-8} = 0$, $\varphi \approx 1.23205 < 1.2321$.
	
	Furthermore, $G$ has at most $\varphi^{\abs{V}}$ preferred extensions.
\end{theorem}

% \openq{I think you might be able to remove the overhead for checking maximality. The gist is, we want to show for all $S$ that are enumerated, if $S \subset T$ and $T$ is also enumerated, then all admissible extensions $U$ such that $S \subseteq U \subseteq T$ are also enumerated. I think with this, one can then remove the overhead from checking maximality.}

% \openq{Could we get better results by directly handling 2-cycles. It just sounds really messy}

% \openq{Is it useful to use measure and conquer? We believe likely not, given how strange the effects of lowering the degree of a vertex is (e.g: a vertex with lower in-degree than out-degree behaves differently from one with higher in-degree than out-degree. Also we do cases based on global properties of vertices, not local properties)}

% \todo{I suspect the complexity is far better than the worst case in practice. Would be nice to have a proof of concept.}

% Is average degree useful?
\section{Bounds on number of preferred extensions}
\label{bounds_sect}
In this section, we collect our results bounding the number of preferred extensions.

\subsection{Bounds on general directed graphs}
A tight upper bound is $O(3^{\frac{\abs{V}}{3}})$ \cite{multipleviewpoints}. This bound is easily attainable since preferred extensions coincide with maximal independent sets(MIS) in graphs where every edge is in a 2-cycle.

We can also enumerate them in $\ostar(3^{\frac{\abs{V}}{3}})$. We note each MIS has a single maximal admissible subset. We then take a branching algorithm for MIS\cite{sergethesis}, allowing us to apply the Maximal Subset Collation subroutine (\autoref{alg:collation}) to remove the non preferred extensions without additional overhead. Hence, all preferred extensions can be enumerated in $\ostar(3^{\frac{\abs{V}}{3}})$.

\subsection{Parameterizing by Resolution Order}
We will give bounds based on $r$, the proportion of vertices that are in at least one 2-cycle.

\subsubsection{Lower Bound}
An undirected 3-cycle has 3 preferred extensions. Hence, applying the Oriented Translation to it, we obtain an oriented structure with 6 vertices and 3 preferred extensions.

Our construction for lower bounding will be to include as many undirected 3-cycles as possible and then include as many oriented translations of 3-cycles as possible. We can include $\frac{r\abs{V}}{3}$ undirected 3-cycles and $\frac{(1-r)\abs{V}}{6}$ oriented translations, obtaining an AF with $\Omega((3^{\frac{r}{3}}3^{\frac{1-r}{6}})^{\abs{V}}) \approx ((1.44^r1.2^{1-r})^{\abs{V}})$ preferred extensions.

\subsubsection{Upper Bound}
% TODO. GRAPH AND INDICATE WHEN OPTIMALITY SWITCHES.
There are 3 different upper bounds that are all optimal in a different range. See also \autoref{fig:run-time}.
\begin{itemize}
	\item $(\varphi^{2r}\varphi^{1-r})^{\abs{V}} \approx O((1.5180^r 1.2321^{1-r})^{\abs{V}})$ where $\varphi$ is the unique positive root of $1 - x^{-1} - x^{-8}$. This bound is obtained from using the Oriented Translation along with \autoref{betterorientedbound}. This is best for $r$ up to around 0.6684. %0.72.
	\item $\ostar(((1+2^{\frac{1}{4}}-\frac{1}{\sqrt{2}})^r(2-\frac{1}{\sqrt{2}})^{1-r})^{\abs{V}})$ $\approx$  $\ostar((1.4822^r1.2929^{1-r})^{\abs{V}})$, the bound from our 2-phase Monotone Local Search. This is best for a small range where $0.6685 \le r \le 0.8004$.
	\item  $3^\frac{\abs{V}}{3}$. This is best for $r \geq 0.8005$.
\end{itemize}

\section{Conclusion}
\label{sect:concl}

We again note that the concept of an admissible (resp. preferred) extension has also been studied as a \emph{semikernel}\cite{NeumannLara71} (resp. maximal \emph{semikernel}) in graph theory. Hence our result may be interpreted as a combinatorial upper bound on the number of maximal \emph{semikernels}, parameterized by the proportion of vertices in at least one 2-cycle.

\subsection*{Acknowledgment}
We thank Oliver Fisher for fruitful discussions and collaboration on preliminary results in the early stages of this work.

\newpage
\begin{appendices}
\section{Background on Analysis of Branching Algorithms}
    \label{sect:background-branching}
    This section provides the required background for the analysis of branching algorithms used in this paper. All results here can be found in standard textbooks. We mostly follow Fomin \& Kratsch \cite{exactexpalgos}:

For the analysis of branching algorithms, we use Measure \& Conquer \cite{FominGK09} with relatively simple measures.
A \emph{measure} $\mu$ is a function assigning a non-negative number to each instance $I$.
For each instance, either the instance is directly solvable in polynomial time or our algorithm will apply a polynomial time branching rule to generate multiple sub-instances with smaller measure that are then recursed on.

The following notation (which is used in standard textbooks on exact exponential algorithms \cite{exactexpalgos} and parameterized algorithms \cite{ParameterizedAlgorithms}) will simplify the discussion of these algorithms.

\begin{definition}[Branching Vector \cite{exactexpalgos}]
% TODO: remove the suitable, sounds iffy.
Let $\mu$ be a measure.
Let $b$ be a branching rule that for any input instance, say $I$, branches into $r$ instances with measures $\mu(I)-t_1, \mu(I)-t_2, \ldots, \mu(I)-t_r$ such that for all $i$, $t_i > 0$. Then we call $\mathbf{b} = (t_1,t_2, \ldots, t_r)$ the \emph{branching vector} of branching rule $b$.
\end{definition}
%All branching rules in this report satisfy the requirements in the above definition (though the definition of a "suitable" input instance will vary between rules).
The following lemma from \cite{exactexpalgos} forms the basis for the analysis of our branching algorithms.

\begin{lemma}[\cite{exactexpalgos}]
Let $b$ be a branching rule with branching vector $(t_1,t_2,\ldots t_r)$. Then the running time of the branching algorithm using only branching rule $b$ is $\ostar(\alpha^{\mu(I)})$ where $\alpha$ is the unique positive real root of:
\[x^n - x^{n-t_1} - x^{n-t_2} - \ldots - x^{n-t_r} = 0 \]
We call $\alpha$ the \emph{branching number} of the branching vector $\mathbf{b}$.
\end{lemma}

In \cite{exactexpalgos}, it is further noted that if an algorithm has multiple branching rules, then its running times is $\ostar(\alpha^n)$ where $\alpha = \max_i \alpha_i$. We will actually need the following simple extension, which has been simplified from \cite{sergethesis} and restated to use our notation.

\begin{lemma}[Combine Analysis Lemma \cite{sergethesis}]
    \label{combineanalysislemma}
Let $A$ be an algorithm for a problem which for each instance, say $I$, either directly solves the instance in $\ostar(\alpha_0^{\mu(I)})$ time or after polynomial time, applies one of $r$ branching rules $b_i$, the $i$-th of which has branching number $\alpha_i$. Then $A$ has running time $\ostar(\alpha^{\mu(I)})$ where $\alpha = \max_{0\le i\le r} (\alpha_i)$.
\end{lemma}

For our enumeration bounds we will need the following variant of the above lemma.

% TODO: check this "in all our applications" is true.
First, we define the search tree of a branching algorithm to be the tree formed by the recursive calls, with the leaves being cases that can be solved directly. Generally, and in all our applications, for enumeration algorithms the number of leaves is an upper bound on the number of objects being enumerated (equality may not hold as the same object may appear as a leaf multiple times).
% TODO, maybe include: We remark that in all our applications, removing duplicates can easily be done through any type of dictionary structure and incurs a polynomial overhead.

\begin{lemma}[Combine Analysis Lemma for Enumeration]
    \label{combineanalysislemmaforenumeration}
Let $A$ be an algorithm for a problem which for each instance, say $I$, either directly solves the instance with a branching algorithm that generates at most $\alpha_0^{\mu(I)}$ leaves or applies one of $r$ branching rules $b_i$, the $i$-th of which has branching number $\alpha_i$. Then the search tree for applying $A$ to $I$ has at most $\alpha^{\mu(I)}$ leaves where $\alpha = \max_i (\alpha_i)$.
\end{lemma}

\section{Background on Monotone Local Search}
\label{sect:background-mls}
Monotone Local Search \cite{monotonelocalsearch} is a technique to design exponential-time algorithms for subset problems based on single-exponential parameterized algorithms.
We give an exposition of the corresponding enumeration algorithms from \cite{monotonelocalsearch}, and the extension in \cite{multivariatesubroutines}.

An \emph{implicit set system} is a function $\Phi$ that takes as input a string $I \in \{0,1\}^*$ and outputs a set system $(U_I, \cF_I)$, where $U_I$ is a universe  and $\cF_I$ is a collection of subsets of $U_I$.  The string $I$ is referred to as an \emph{instance}  and we denote by $|U_I| = n$ the size of the universe and by $|I|=\bitsize$ the size of the instance. 
We assume that $\bitsize\ge n$. An implicit set system $\Phi$ is \emph{polynomial time computable} if: (a) there exists a polynomial time algorithm that given $I$ produces $U_I$, and (b) there exists a polynomial time algorithm that given $I$, $U_I$ and a subset $S$ of $U_I$ determines whether $S \in \cF_I$.

Let $c,b \geq 1$ be real valued constants and $\Phi$ be an implicit set system.
Then $\Phi$ is \emph{($b,c$)-uniform} if for every instance $I$, set $X \subseteq U_I$, and integer $k \leq n - |X|$, the cardinality of the collection
\begin{align*}
\cF_{I,X}^k = \{ S \subseteq U_I \setminus X : |S| = k \text{ and } S \cup X \in \cF_I \}
\end{align*}
is at most $b^{n - |X|} c^k n^{O(1)}$. If there is an algorithm that given any instance $I$, set $X \subseteq U_I$, and integer $k \leq n - |X|$, enumerates all elements of $\cF_{I,X}^k$
in time $b^{n - |X|}c^k \bitsize^{O(1)}$, then $\Phi$ is said to be \emph{efficiently ($b,c$)-uniform}.
In this case, the algorithm from \cite{multivariatesubroutines} enumerates $\cF_I$ in $\ostar\left(\left( 1 + b - \frac{1}{c} \right)^{n+o(n)}\right)$ time.

\begin{theorem}
	\label{thm:appendixmls}
	Let $c,b \geq 1$ and $\Phi$ be a polynomial time computable implicit set system.
	
	If $\Phi$ is $(b,c)$-uniform, then $|\cF_I| \leq \left(  1 + b - \frac{1}{c} \right)^n n^{O(1)}$ for every instance $I$.
	
	If $\Phi$ is efficiently $(b,c)$-uniform, then there is an algorithm that, given $I$ as input, enumerates $\cF_I$ in time $\left( 1 + b - \frac{1}{c} \right)^{n+o(n)} \bitsize^{O(1)}$.
\end{theorem}

\noindent
The original paper on Monotone Local Search \cite{monotonelocalsearch} proved the theorem for $b=1$.

\section{Oriented Translation}
\label{sect:orientedtranslation}
In this section we present the construction we use to transform directed graphs into oriented graphs (see \autoref{outline:orientedtranslation}) while preserving the preferred extensions.

This will be useful for lifting complexity results from directed graphs to oriented graphs. It is also used for automatically converting oriented graph algorithms into algorithms for general graphs parameterized by resolution order.

To remove special cases, we first present a translation that eliminates self loops from an AF. The only purpose of self loops is to prevent vertices from being included in admissible sets. Hence, we can replace self loops with attacks from an undefendable vertex.

\subsection*{Loopless Translation}
\begin{definition}[Loopless Translation]
\label{looplesstranslation}
Let $G = (V,E)$ be an arbitrary AF. The loopless translation of $G$ is the graph $G' = (V', E')$ where:
\begin{itemize}
\item $V' = V \cup \{l_1, l_2, l_3\}$.
\item $E' = (E \setminus \{\edge{v, v} : v \in V \}) \cup \{\edge{l_1, v} : v \in V, \edge{v,v} \in E \} \cup \{\edge{l_1, l_2}, \edge{l_2, l_3} , \edge{l_3, l_1}\}$
\end{itemize}
\end{definition}
\noindent
An example translation of $V = \{v_1, v_2\}, E = \{\edge{v_1,v_1}, \edge{v_1,v_2}\}$ can be found in \autoref{looplesstranslation_example}.
\begin{figure}[tb]
	\centering
	\begin{tikzpicture}
	\tikzstyle{nd}=[circle,draw=black,thick,fill=gray!10,minimum size=20pt,inner sep=0pt]
	\node[nd](l1) at (0,0) {$l_1$};
	\node[nd, right=1cm of l1](l2){$l_2$};
	\node[nd, above right=0.5cm and 0.5cm of l1](l3){$l_3$};
	\node[nd, below = 1cm of l1](v1){$v_1$};
	\node[nd, right = 1cm of v1](v2){$v_2$};
	\begin{scope}
		\tikzset{edge/.style = {->,> = latex'}}
		\draw[edge] (l1) -- (l2);
		\draw[edge] (l2) -- (l3);
		\draw[edge] (l3) -- (l1);
		
		\draw[edge] (l1) -- (v1);
		\draw[edge] (v1) -- (v2);
	\end{scope}
	\end{tikzpicture}
	\caption{Loopless Translation of $V = \{v_1, v_2\}, E = \{\edge{v_1, v_1}, \edge{v_1, v_2} \}$}
	\label{looplesstranslation_example}
\end{figure}
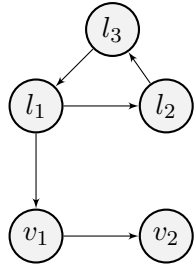
\begin{lemma}
The admissible extensions of $G$ are exactly the admissible extensions of the Loopless Translation $G'$ of $G$.
\end{lemma}
\begin{proof}
The odd cycle among $\{l_1, l_2, l_3\}$ ensures none of $l_1, l_2, l_3$ are in any admissible extension of $G'$. Hence $l_1$ can not be attacked in any admissible extension of $G'$. Hence, no vertex with a self loop in $G$ is in any admissible extension of $G'$.

Now it is straight forward to check $S \subseteq V(G') \setminus \{l_1, l_2, l_3\}$ is an admissible extension of $G'$ if and only if it is an admissible extension of $G$.
\end{proof}

We summarize the result as follows.

\begin{lemma}
\label{lemma:looplesstranslation}
Let $G$ be an AF. Then there is a linear time algorithm that transforms $G$ into an AF $G'$, with $\abs{V(G')} = \abs{V(G)} + 3$, that has no self loops.

Furthermore there is an inclusion preserving bijection between the admissible extensions of $G$ and the admissible extensions of $G'$ that can be applied in linear time.
\end{lemma}

This is why we assume in all other results that graphs have no self loops.

\subsection*{Simple Oriented Translation}
For pedagogical purposes we first present an oriented translation that doubles the number of vertices. The full oriented translation algorithm follows from noting that only vertices in a 2-cycle need to be duplicated.

Our strategy is to replace each 2-cycle with a 4-cycle. However, a bit more care is necessary to ensure a bijection between preferred extensions in both graphs.
\begin{definition}[Simple Oriented Translation]
Let $G = (V,E)$ be an arbitrary AF. First apply the Loopless Translation (\autoref{looplesstranslation}) to $G$. Arbitrarily order the vertices in $V$ with a total order $<$. For simplicity, relabel $V$ such that $V = \{v_1,v_2,\ldots v_n \}$ with $v_1 < v_2 < \ldots < v_n$. Let $F$ be the set of edges in $G$ that are in a 2-cycle.
Then the simple oriented translation of $G$ is the graph $G' = (V',E')$ where:
% TODO: STILL NOT SURE WHAT THE NEATEST WAY OF HANDLING THE NON-EXISTENCE OF w_i IS.
\begin{itemize}
\item $V' = \{v_i : 1 \leq i \leq n \} \cup \{w_i : 1 \leq i \leq n \}$
\item $E'$ is the union of the following sets:
\begin{itemize}
\item 4-cycle edges: $\{ \edge{v_i,v_j},\edge{v_j,w_i}, \edge{w_i,w_j}, \edge{w_j,v_i} : \edge{v_i,v_j} \in F \land i < j \}$.
\item Unidirectional edges: $\{ \edge{v_i, v_j}, \edge{v_i,w_j}, \edge{w_i,v_j}, \edge{w_i,w_j} : \edge{v_i,v_j} \in E, \edge{v_i, v_j} \notin F  \}$.
\end{itemize}
\end{itemize}
\end{definition}
\noindent
An example translation of $V = \{v_1, v_2, v_3\}, E = \{\edge{v_1,v_2}, \edge{v_2,v_1}, \edge{v_2,v_3}\}$ can be found in \autoref{simpleorientedtranslation_example}.

\begin{figure}[tb]
	\centering
	\begin{tikzpicture}
	\tikzstyle{nd}=[circle,draw=black,thick,fill=gray!10,minimum size=20pt,inner sep=0pt]
	\node[nd](v1) at (0,0) {$v_1$};
		\node[nd, below = 2cm of v1](w1) {$w_1$};
	\node[nd, right = 2cm of v1](v2) {$v_2$};
		\node[nd, below = 2cm of v2](w2) {$w_2$};
	\node[nd, right = 2cm of v2](v3) {$v_3$};
		\node[nd, below = 2cm of v3](w3) {$w_3$};
	\begin{scope}
		\tikzset{edge/.style = {->,> = latex'}}
		\draw[edge] (v1) -- (v2);
		\draw[edge] (v2) -- (w1);
		\draw[edge] (w1) -- (w2);
		\draw[edge] (w2) -- (v1);
		
		\draw[edge] (v2) -- (v3);
		\draw[edge] (v2) -- (w3);
		\draw[edge] (w2) -- (v3);
		\draw[edge] (w2) -- (w3);
	\end{scope}
	\end{tikzpicture}
	\caption{Simple Oriented Translation of $v_1 \leftrightarrow v_2 \rightarrow v_3$}
	\label{simpleorientedtranslation_example}
\end{figure}
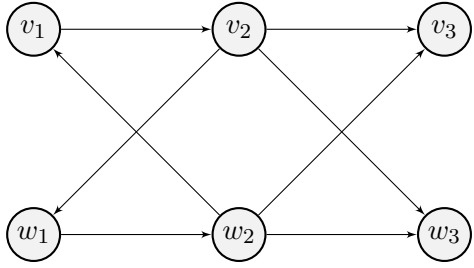

We need to show that $G'$ has a few nice properties. First:

\begin{lemma}
$G'$ is oriented.
\end{lemma}
\begin{proof}
    This is clear from a case analysis on the edges of $G'$.
\end{proof}

Next, we want to show this translation nicely preserves preferred extensions. We first need a lemma controlling the structure of preferred extensions in $G'$.
\begin{lemma}
\label{inclusivepropertylemma}
Let $D$ be a preferred extension of $G' = (V', E')$. For any $i$, if $v_i$ or $w_i$ are in $D$, then both $v_i$ and $w_i$ are in $D$.
\end{lemma}

\begin{proof}
    First we note $N(v_i) = N(w_i)$ so if $v_i$ is conflict-free with respect to $D$, then $w_i$ is conflict-free with respect to $D \cup \{v_i\}$.

    We then note that every vertex attacking $w_i$ either also attacks $v_i$ or is attacked by $v_i$. Hence if $v_i$ is acceptable with respect to $D$, then so is $w_i$.

    Therefore if $v_i \in D$ is admissible, then so is $w_i$ with respect to $D$. Hence, by maximality, $w_i \in D$. This clearly also applies if $w_i \in D$.
\end{proof}

Now the main result:
\begin{lemma}
Let $G$ be an AF. Let $G'$ be the oriented AF obtained from applying the simple oriented translation on $G$. Let $\psi$ be the following function from subsets of $V(G)$ to subsets of $V(G')$:
\[ \psi(S) = \{ v_i : v_i \in S \} \cup \{ w_i : v_i \in S \} \]
Then $\psi$ is a bijection between the preferred extensions of $G$ and the preferred extensions of $G'$.
\end{lemma}

\begin{proof}
    It is clear that $\edge{v_i, v_j} \in G \iff \{v_i, w_i\}$ attacks both $v_j$ and $w_j$. Hence $\psi$ is a bijection between the preferred extensions of $G$ and the preferred extensions of $G'$ where $(v_i \in S \lor w_i \in S) \implies (v_i \in S \land w_i \in S)$. However \autoref{inclusivepropertylemma} states that every preferred extension of $G'$ satisfies this property.
\end{proof}

\subsection*{Full Oriented Translation}
We have shown that the simple oriented translation does biject preferred extensions. We now optimize the translation which is necessary for our parameterized results.

Namely, if $v_i \in V$ is not in any 2-cycles, then it is easy to see that $v_i$ and $w_i$ have the exact same in-neighbors and out-neighbors in $G'$. Hence, $v_i$ and $w_i$ are essentially identical and merging $v_i$ and $w_i$ into a single vertex still preserves all the above lemmas. Formally, an explicit construction of the oriented translation is:

\begin{definition}[Oriented Translation]
    \label{orientedtranslation-construction}
Let $G = (V,E)$ be an arbitrary AF.  First apply the Loopless Translation (\autoref{looplesstranslation}) to $G$. Arbitrarily order the vertices in $V$ with a total order $<$. For simplicity, relabel $V$ such that $V = \{v_1,v_2,\ldots v_n \}$ with $v_1 < v_2 < \ldots < v_n$. Let $S$ be the vertices in $G$ that are in at least one 2-cycle. Let $F$ be the set of edges in $G$ that are in a 2-cycle.

Then the oriented translation of $G$ is the graph $G' = (V',E')$ where:
\begin{itemize}
\item $V' = \{v_i : 1 \leq i \leq n \} \cup \{w_i : v_i \in S \}$
\item $E'$ is the union of the following sets:
\begin{itemize}
\item 4-cycle edges: $\{ \edge{v_i,v_j},\edge{v_j,w_i}, \edge{w_i,w_j}, \edge{w_j,v_i} : \edge{v_i,v_j} \in F \land i < j \}$. We note that all these vertices exist as $v_i, v_j \in S$ since there is a 2-cycle between $v_i$ and $v_j$.
\item Unidirectional edges: $\{ \edge{v_i, v_j}, \edge{v_i,w_j}, \edge{w_i,v_j}, \edge{w_i,w_j} : \edge{v_i,v_j} \in E, \edge{v_i, v_j} \notin F  \}$. Leave out of $E'$ all edges in this category that have a non-existent endpoint (this may happen as $w_i$ does not necessarily exist for all $i$).
\end{itemize}
\end{itemize}
\end{definition}

From this construction, it is clear the lemmas for the Simple Oriented Translation still hold in slightly modified form:

\begin{lemma}
$G'$ is oriented.
\end{lemma}
\begin{lemma}
Let $G$ be an AF. Let $G'$ be the oriented AF obtained from applying the oriented translation to $G$. Let $\psi$ be the following function from subsets of $V(G)$ to subsets of $V(G')$:

\[ \psi(S) = \{ v_i : v_i \in S \} \cup \{ w_i : v_i \in S \land w_i\text{ exists} \} \]

Then $\psi$ is a bijection between the preferred extensions of $G$ and the preferred extensions of $G'$.
\end{lemma}

The following theorem summaries the main results of this section:
\begin{theorem}
    \label{orientedtranslationthm}
Let $G$ be an AF. Let $r(G)$ be the number of vertices in $G$ that are in at least one 2-cycle.

There is a linear time algorithm that transforms any AF $G$ into an oriented AF $G'$ with $\abs{V(G')} = \abs{V(G)} + r(G) + 3$
such that
there is a
bijection between the preferred extensions of $G$ and the preferred extensions of $G'$ that can be applied in linear time.
\end{theorem}
\noindent
The 3 in $\abs{V(G')} = \abs{V(G)} + r(G) + 3$ comes from the Loopless Translation (see \autoref{lemma:looplesstranslation}).

\section{Maximal Subset Collation}
\label{sect:maximalsubsetcollation}
Recall we define the problem as:
\pbDefOpt{Maximal Subset Collation}{Graph $G$, $c$ pairs $(S_i, C_i)$ where for each $i$, $S_i \subseteq V(G)$ and $C_i$ is a set containing only maximal admissible subsets of $S_i$}{A set containing all maximal elements of $\bigcup_{i=1}^{c} C_i$}
The key to our algorithm for Maximal Subset Collation is the following lemma:

\begin{lemma}
    \label{collationpruninglemma}
    Suppose $T \in C_i$ but $T \subsetneq U \in \bigcup_{i=1}^{c} C_i$. Then $T$ is the unique maximal admissible subset of $U \cap S_i$.
\end{lemma}
\begin{proof}
    First, we note that $U \cap S_i$ has a unique maximal admissible subset due to \autoref{lemma:DAGmaximaladmissible} ($U$ is admissible, hence conflict-free, hence $U \cap S_i$ induces a DAG).

    It is clear $T \subseteq U \cap S_i$. Then, since $T$ is a maximal admissible subset of $S_i$, $T$ must be the unique maximal admissible subset of $U \cap S_i$ as required.
\end{proof}

Now, our algorithm for Maximal Subset Collation is straight-forward with pseudocode in \autoref{alg:collation}. \autoref{collationpruninglemma} guarantees all non-maximal subsets in $\bigcup_{i=1}^{c} C_i$ are removed. By \autoref{lemma:DAGmaximaladmissible}, finding the Unique Maximal Admissible Subset of $U \cap S_i$ can be done in polynomial time. Using a hashmap or balanced binary search tree to store $R$ (in \autoref{alg:collation}) yields the following result:

\begin{lemma}
    \label{appendix:collationlemma}
    % TODO: check the ostar?
    Maximal Subset Collation can be solved in $O(c\sum_{i=1}^{c} \abs{C_i}\cdot\poly(\abs{V}))$ time. 
\end{lemma}

\begin{algorithm}[tb]
\caption{Maximal Subset Collation}
\label{alg:collation}
\begin{algorithmic}[1]
\Require For all $i \in \{1,\ldots,c\}$, $C_i$ only contains maximal admissible subsets of $S_i$
\Ensure $R = \{S : S \text{ is maximal in } \bigcup_{i=1}^{c} C_i \}$
\Function{Maximal Subset Collation}{$(S_1, C_1), \ldots, (S_c, C_c)$}
	\State $R \gets \bigcup_{i=1}^{c} C_i$
	\ForAll{$U \in R$}
		\ForAll{$i \in \{1, \ldots c\}$}
			\State $M \gets \text{The unique maximal admissible subset of } U \cap S_i$
			\If{$M \neq U$}
				\State $R \gets R \setminus \{M\}$
			\EndIf
		\EndFor
	\EndFor
	\State \Return $R$
\EndFunction
\end{algorithmic}
\end{algorithm}

\section{Detailed Analysis of our \texorpdfstring{$\ostar(2^{\mu(I)})$}{parameterized} algorithm for Maximal Admissible Subset Enumeration}
\label{sect:MASE-cases}
Recall we have the following cases and apply them in the listed order:

\begin{center}\small
	\newcolumntype{L}[1]{>{\hsize=#1\hsize\raggedright\arraybackslash}X}

    \begin{tabularx}{1\textwidth}{L{0.15} L{1.3} L{0.8}}
\toprule
Case & Requirement to apply & Running Time \\
\midrule%\hline
Base & $S$ is conflict-free. & Solves in $O^{*}(1)$ time, returns $\leq 1$ set. \\
%\hline
1 & $G[S]$ is oriented with maximum total degree $\leq 2$. & Branching vector $(1,1)$, branching number $2$. \\
%\hline
2 & There is a 2-cycle in $G[S]$. & Branching vector $(1,1)$, branching number $2$. \\
%\hline
3 & $G[S]$ has maximum total degree $\ge 4$. & Branching vector $(2,\frac{1}{2})$, branching number $\approx 1.91$. \\
%\hline
4 & $G[S]$ contains a vertex with total degree 3. & Branching vector $(1,\frac{3}{2})$, branching number $\approx 1.76$. \\
\bottomrule
\end{tabularx}
\end{center}

These requirements are exhaustive, hence there will always be at least one applicable rule in any instance.

\subsection{Base Case - \texorpdfstring{$S$}{S} is conflict-free}
\label{subsect:MASE_BaseCase}

By \autoref{lemma:DAGmaximaladmissible}, $S$ has a unique maximal admissible subset, $T$, which can be found in polynomial time. Return this subset if $\abs{S \setminus T} \leq k$.

\subsection{Case 1 - \texorpdfstring{Oriented, degree $\leq 2$ graph}{Oriented graph, degrees at most 2}}
\label{subsect:MASE_Case1}
As the maximum total degree in $G[S]$ is $\leq 2$, ignoring arc directions, $G[S]$ decomposes into disjoint paths and cycles. Since \textbf{(Undefendable)} is not applicable, each non-isolated vertex in $G[S]$ must have an in-neighbor. Hence, we can further deduce that $G[S]$ decomposes into disjoint directed cycles and isolated vertices.

Let $v \in S$ be any non-isolated vertex in $G[S]$ (such a vertex exists as $S$ is not conflict-free since the base case is not applicable).
We do a 2-way branch on whether $v$ is in the admissible extension.

In one branch, we assume $v$ is in the admissible extension. Then neither of $v$'s neighbors are in the admissible extension. Hence our new instance $I' = (S', k', b')$ has:
\begin{itemize}
    \item $S' = S \setminus N(v)$.
    \item $k' = k-2$. The two neighbors of $v$ are distinct as $G[S]$ is oriented.
    \item $b' \leq b$.
\end{itemize}
and hence $\mu(I') \leq \mu(I) - \frac{2}{2} = \mu(I) - 1$.

\vspace{1mm}
In the other branch, we remove $v$ from $S$.
After this, the directed cycle that $v$ was a part of is now a directed path.
This path has at least two vertices as, by assumption, $G[S]$ is oriented.
Hence, after removing $v$, we can apply the rule \textbf{(Undefendable)} to prune at least one more vertex, say $w$, from $S$.
Hence our new instance $I' = (S', k', b')$ has:
\begin{itemize}
    \item $S' = S \setminus \{v,w\}$.
    \item $k' = k-2$.
    \item $b' \leq b$.
\end{itemize}
and hence $\mu(I') \leq \mu(I) - \frac{2}{2} = \mu(I) - 1$.

% TODO: maybe note what actually is happening, we're simplifying out the entire cycle.
Hence we get a branching vector of $(1, 1)$ with branching number $2$.

\subsection{Case 2 - 2-cycle exists}
\label{subsect:MASE_Case2}
Let $v,w \in S$ be two vertices such that $\edge{v,w}, \edge{w,v} \in E$. There are two cases depending on whether $\{v,w\}$ is a maximal weakly connected component in $G[S]$.

\textbf{Case 1:} If $\{v,w\}$ is a maximal weakly connected component, then we branch on whether $v$ is in the admissible extension.

In the branch where we include $v$ we must remove $w$. Hence our new instance has $S' = S \setminus \{w\}, k' = k-1$.

In the other branch we remove $v$. Hence our new instance has $S' = S \setminus \{v\}, k' = k-1$.

In both cases $v$ and $w$ are no longer in any 2-cycles and hence $b' = b-2$.
Therefore, in both cases our new instance has $\mu(I') = \mu(I) - \frac{1}{2} - \frac{2}{4} = \mu(I) - 1$
and we get a branching vector of $(1,1)$ with branching number $2$.

\textbf{Case 2:} If $\{v,w\}$ is not a maximal weakly connected component, then one of these vertices has another neighbor. Without loss of generality, we assume $v$ has at least two distinct neighbors. Again, we branch on whether we include $v$ in the admissible extension.

In the branch where we include $v$ we must remove $N(v)$.
Furthermore, since $v$ is no longer in any 2-cycles and $w$ has been removed, $b$ decreases by at least 2.
Hence our new instance $I' = (S', k', b')$ has:
\begin{itemize}
    \item $S' = S \setminus N(v)$.
    \item $k' = k - \abs{N(v)} \leq k - 2$.
    \item $b' \leq b - 2$.
\end{itemize}
and hence $\mu(I') \leq \mu(I) - \frac{2}{2} - \frac{2}{4} = \mu(I) - \frac{3}{2}$.

% and $k \gets k-\abs{N(v)} \leq k - 2$. Furthermore, since $v$ is no longer in any 2-cycles and $w$ has been removed, $b$ decreases by at least 2.

\vspace{1mm}
In the branch where we exclude $v$ we remove $v$ from $S$.
Hence in our new instance $I' = (S', k', b')$:
\begin{itemize}
    \item $S \gets S \setminus \{v\}$.
    \item $k' = k - 1$.
    \item $b' \leq b - 1$ as $v$ is in a 2-cycle.
\end{itemize}
and hence $\mu(I') \leq \mu(I) - \frac{1}{2} - \frac{1}{4} = \mu(I) - \frac{3}{4}$.

Hence we get a branching vector of $(\frac{3}{2}, \frac{3}{4})$ with branching number $\approx 1.9 < 2$.

\subsection{Case 3 - \texorpdfstring{Total Degree $\geq 4$ exists}{Total Degree at least 4 exists}}
\label{subsect:MASE_Case3}

Let $v$ be any vertex in $G[S]$ with total degree $\geq 4$. So $\abs{N(v)} \geq 4$ as we are assuming $G$ has no self-loops and after case 2, we may also assume $G[S]$ is oriented.

We branch on whether to include $v$ in our admissible subset. If we include $v$, we must remove its $\geq 4$ neighbors. Otherwise we remove $v$.
In the first case the new instance has $k' \leq k - 4$. In the second case the new instance has $k' = k - 1$.

Hence we get a branching vector of $(2,\frac{1}{2})$ with branching number $\approx 1.91 < 2$.

\subsection{Case 4 - Degree 3 exists}
\label{subsect:MASE_Case4}

We need the following lemma on the graph structure at this point. In this branching rule the conditions are satisfied as otherwise an earlier branching rule would have been applicable.

\begin{lemma}
\label{MASE:lemma:deg12exists}
Suppose $G[S]$ is oriented, \textbf{(Undefendable)} is not applicable to $S$ and the maximum total degree in $S$ is 3. Then there exists a vertex in $G[S]$ with in-degree 1 and out-degree 2.
\end{lemma}
\begin{proof}
Let $u$ be any vertex with total degree 3 in $G[S]$. First, $u$ cannot have in-degree 0, otherwise \textbf{(Undefendable)} would be applicable. Hence if $u$ has higher out-degree than in-degree, then $u$ must have in-degree 1 and out-degree 2.

Otherwise, as the sum of in-degrees in $G[S]$ is equal to the sum of out-degrees, there must be a vertex, $v$, in $G[S]$ with higher out-degree than in-degree. Since there are no vertices with degree $\geq 4$ and no vertices with in-degree 0 and out-degree $\neq 0$, necessarily $v$ must have in-degree 1 and out-degree 2.

Hence, in either case there exists a vertex in $G[S]$ with in-degree 1 and out-degree 2, as required.
\end{proof}

Let $v$ be any vertex in $G[S]$ with in-degree 1 and out-degree 2. Let $u$ be the unique vertex in $G[S]$ attacking $v$. We branch on whether to include $u$.

If we include $u$, then in our new instance $I' = (S', k', b')$ we have:
\begin{itemize}
    \item $S' = S \setminus N(u)$.
    \item $k' = k - \abs{N(u)} \leq k - 2$. We note $u$ has at least $1$ in-neighbor (as \textbf{(Undefendable)} is not applicable) and this in-neighbor is distinct from $v$ as $G[S]$ is oriented.
    \item $b' = b$.
\end{itemize}
and hence $\mu(I') \leq \mu(I) - \frac{2}{2} = \mu(I) - 1$.

\vspace{1mm}

If we do not include $u$, then after removing $u$ from $S$, $v$ has in-degree 0. Hence we can apply \textbf{(Undefendable)} to the two vertices that $v$ attacks, removing those as well. These 3 vertices are all distinct as $G[S]$ is oriented. Hence in our new instance $I' = (S', k', b')$ we have:
\begin{itemize}
    \item $S' = S \setminus N(v)$.
    \item $k' = k -3$.
    \item $b' = b$.
\end{itemize}
and hence $\mu(I') \leq \mu(I) - \frac{3}{2}$.

Hence we get a branching vector of $(1,\frac{3}{2})$ with branching number $\approx 1.76 < 2$.

\section{Complexity Lower Bounds}
\label{sect:complexitylowerbounds}
    In this section we will first show, assuming P != NP, that there is no output-polynomial time algorithm (an algorithm where the running time is upper bounded by a polynomial in input size plus output size) for preferred or admissible enumeration.

    We also show our parameterized algorithm is tight, assuming the Strong Exponential Time Hypothesis (SETH). For this, we instead focus on the decision version of the problem:

\pbDef{Admissible Removal (AR)}{Graph $G$, set $S \subseteq V(G)$, integer $k$}{$k$}{Does there exist an admissible set $T \subseteq S$ such that $\abs{S \setminus T} \leq k$}
    
    We show, assuming SETH, there is no $\ostar(2^{(1-\varepsilon)\mu(I)})$ algorithm for AR. Separately, we will also show there is no $\ostar(2^{(1-\varepsilon)k})$ algorithm for AR. We briefly note a straight-forward 2-way branching $\ostar(2^{k})$ algorithm exists for AR and MASE.

\subsection{Required Background}
    We will be using the following form of the Strong Exponential Time Hypothesis (SETH):
\begin{hypothesis}[Strong Exponential Time Hypothesis]
There is no $O(2^{(1-\varepsilon)n})$ algorithm for SAT for any $\varepsilon > 0$. Here, $n$ denotes the number of variables in the input instance. 
\end{hypothesis}

The primary tool used to relate known complexity results to problems for Argumentation Frameworks is a translation from Boolean formulae in Conjunctive Normal Form (CNF) to AFs that has been called the \textit{standard translation} (see, e.g., \cite{ArgumentationBookComplexityCh}).
We will need a slight variation of the \textit{standard translation} that was used in \cite{graphtheoreticalstructures}.

% TODO: Technically not the same as the cited source since we add extra clauses. But I think they were meant to too given they "prove" admissibles and preferreds coincide in their construction.
% SG: ok, let's leave it as is.
\begin{definition}[Extended translation \cite{graphtheoreticalstructures}]
Given $\varphi(z_1,\ldots z_n)$, a CNF with clauses \allowbreak $\{C_1, \ldots C_m \}$,
first let $C' = \{C_1, \ldots C_m \} \cup \{(z_i \lor \neg z_i) : 1 \leq i \leq n \}$. The AF $(V,E)$ constituting the extended translation from $\varphi$ has
\begin{itemize}
\item $V = \{\varphi \} \cup \{C'_1,\ldots C'_{m+n} \} \cup \{z_1,\ldots z_n \} \cup \{\neg z_1,\ldots, \neg z_n\} \cup \{A_i : 0 \leq i \leq 2 \}$
\item $\begin{aligned}[t]
        E &= \{\edge{C'_j, \varphi} : 1 \leq j \leq \abs{C'} \} \\
          &\cup \{\edge{z_i, \neg z_i}, \edge{\neg z_i, z_i} : 1 \leq i \leq n\} \\
          &\cup \{\edge{y_i, C'_j} : y_i \in C'_j\} \\
          &\cup \{\edge{\varphi, A_0} \} \\
          &\cup \{\edge{A_0, A_1}, \edge{A_1, A_2}, \edge{A_2, A_0}\} \\
          &\cup \{\edge{A_0, z_i}, \edge{A_0, \neg z_i} : 1 \leq i \leq n \}
    \end{aligned}$
\end{itemize}
\end{definition}
\noindent
An example translation can be found in \autoref{extendedtranslation_example}. The idea is to convert the notion of $y_i$ satisfying $C'_j$ in the CNF into the AF notion of $y_i$ defending $\varphi$ from $C'_j$.

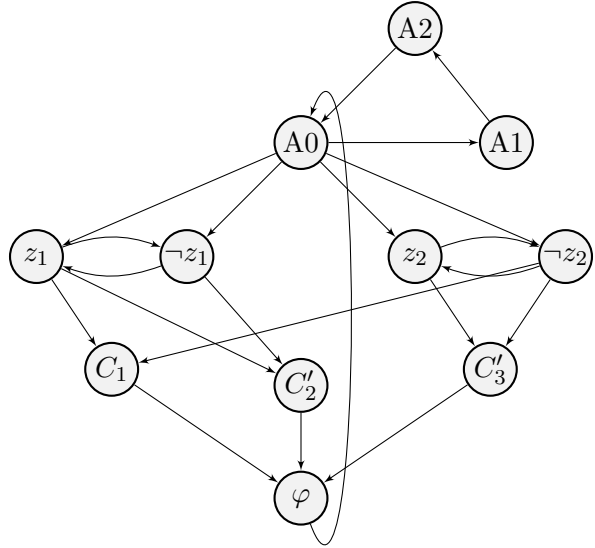
\begin{figure}[tb]
	\centering
	% TODO: FIX THIS GOD AWFUL TIKZ JOB...
	\begin{tikzpicture}
	\tikzstyle{nd}=[circle,draw=black,thick,fill=gray!10,minimum size=20pt,inner sep=0pt]
	\node[nd](A0) at (0,0) {A0};
	\node[nd, right=2cm of A0](A1){A1};
	\node[nd, above right=1cm and 1cm of A0](A2){A2};
	\node[nd, below left=1cm and 3cm of A0](z1) {$z_1$};
		\node[nd, below left=1cm and 1cm of A0](nz1) {$\neg z_1$};
	\node[nd, below right=1cm and 1cm of A0](z2) {$z_2$};
		\node[nd, below right=1cm and 3cm of A0](nz2) {$\neg z_2$};
	\node[nd, below left=2.5cm and 2cm of A0](C1) {$C_1$};
		\node[nd, below=2.5cm of A0](C2) {$C'_2$};
			\node[nd, below right=2.5cm and 2cm of A0](C3) {$C'_3$};
	\node[nd, below=4cm of A0](phi) {$\varphi$};
	\begin{scope}
		\tikzset{edge/.style = {->,> = latex'}}
		\draw[edge] (A0) -- (A1);
		\draw[edge] (A1) -- (A2);
		\draw[edge] (A2) -- (A0);
		\draw[edge] (A0) -- (z1);
		\draw[edge] (A0) -- (nz1);
		\draw[edge] (A0) -- (z2);
		\draw[edge] (A0) -- (nz2);
		\draw[edge] (z1) to [out=20,in=160] (nz1);
		\draw[edge] (nz1) to [out=200,in=-20] (z1);
		\draw[edge] (z2) to [out=20,in=160] (nz2);
		\draw[edge] (nz2) to [out=200,in=-20] (z2);
		\draw[edge] (z1) -- (C1);
		\draw[edge] (nz2) -- (C1);
		\draw[edge] (z1) -- (C2);
		\draw[edge] (nz1) -- (C2);
		\draw[edge] (z2) -- (C3);
		\draw[edge] (nz2) -- (C3);
		\draw[edge] (C1) -- (phi);
		\draw[edge] (C2) -- (phi);
		\draw[edge] (C3) -- (phi);
		\draw[edge] (phi) to [bend left=200](A0);
	\end{scope}
	\end{tikzpicture}
	\caption{Extended Translation of $\{C_1 = (z_1 \lor \neg z_2)\}$}
	\label{extendedtranslation_example}
\end{figure}

In particular we have:

\begin{lemma}
\label{ETstructuretheorem}
Let $G$ be an AF obtained from applying the Extended Translation to a CNF formula $\varphi(z_1,\ldots z_n)$. The nontrivial admissible extensions of $G$ are exactly the sets
\[\{\varphi\} \cup T \cup W\]
where
\begin{itemize}
\item $T$ is a subset of $\{z_i : 1 \leq i \leq n \} \cup \{\neg z_i : 1 \leq i \leq n \}$ that satisfies:
\begin{itemize}
\item for all $1 \leq i \leq n$, $T$ contains either $z_i$ or $\neg z_i$, and
\item the literals in $T$ form a satisfying assignment of $\varphi(z_1,\ldots z_n)$.
\end{itemize}
\item $W$ is either $\emptyset$ or $\{ A_1 \}$.
\end{itemize}
\end{lemma}
\begin{proof}
    \label{proofofETstructuretheorem}
Let $S$ be a nontrivial admissible extension of $G$. In our proof, we argue which vertices may belong to $S$.

\paragraph*{$\mathbf{\{A_0, A_1, A_2, \varphi\}}$.}
As $A_1$ is both adjacent to $A_0$ and the only vertex that defends $A_0$ from $A_2$, we deduce that $A_0 \notin S$. For similar reasons, we deduce that $A_2 \notin S$.

Through a simple chain of deductions, it can be verified that if $\varphi \notin S$ and $A_0 \notin S$, then $S$ must be empty. Hence $\varphi \in S$.

Now $A_1$ is defended by $S$ and only attacks an irrelevant vertex $A_2$. Hence $A_1$ may or may not be in $S$ and its presence or absence does not affect the presence or absence of any other vertices in $S$.

\paragraph*{$\mathbf{\{C'_i\}}$.}
As each $C'_i$ is adjacent to $\varphi$, $C'_i \notin S$ for all $i$.

\paragraph*{$\mathbf{\{z_i,\neg z_i \}}$.}
By adjacency it is impossible that both $z_i$ and $\neg z_i$ are in $S$ for any $i$.

To defend $\varphi \in S$, each vertex corresponding to a clause must be attacked. By construction of the graph, to attack the vertex corresponding to a clause $C'_i$, there must be a vertex in $S$ that represents a literal in $C'_i$. Hence, the vertices in $S$ that correspond to literals form a satisfying assignment of $\varphi(z_1,\ldots z_n)$.

Finally, for all $i$ the clause $(z_i \lor \neg z_i)$ is in $C'$. Hence $S$ must contain either $z_i$ or $\neg z_i$.

\medskip
\noindent
The above shows that the conditions given are necessary. It is easy to see that the conditions given are sufficient. Hence the nontrivial admissible extensions of $G$ have exactly the prescribed form.
\end{proof}

From this we obtain the required relation:
\begin{lemma}
Let $\varphi(z_1,\ldots z_n)$ be a CNF formula and $G$ be the AF obtained by applying the Extended Translation. Then there is a bijection between the non-trivial admissible extensions of $G$ with $W = \emptyset$ and the satisfying assignments of $\varphi$.
\end{lemma}
\begin{proof}
    This is clear from the description in \autoref{ETstructuretheorem}.
\end{proof}

\subsection{No Output-Polynomial Enumeration Algorithm}
\label{sect:nooutputpoly}
% TODO: phrasing here isn't great.
As a first result, we partly justify our decision to only focus on worst case running times by showing, assuming P != NP, that no output-polynomial time algorithm (an algorithm where the running time upper bounded by a polynomial in the input size plus the output size) exists for the enumeration of admissible or preferred extensions.

In \cite{graphtheoreticalstructures}, using the Extended Translation, it is proven that it is NP-complete to determine if an AF has a nontrivial admissible extension or a nontrivial preferred extension. Hence assuming P $\neq$ NP, there is no output-polynomial time algorithm that enumerates the admissible extensions or the preferred extensions of a digraph.
In \cite{KrollPW17}, it was shown that the preferred extensions cannot be enumerated in output-polynomial time, unless P=NP, even for bipartite AFs.

By applying the oriented translation algorithm to the AF obtained from the extended translation, it is straight forward to show that these impossibility results also hold on oriented graphs. Hence:

\begin{theorem}
There is no algorithm that enumerates the admissible or preferred extensions of an AF in output-polynomial time. These results further hold when restricted to the class of oriented graphs.
\end{theorem}

\subsection{Parameterized Complexity Lower Bounds}
In this section we show, assuming SETH, that our parameterized algorithm for MASE is tight for the measure we use. In particular we show there is no $\ostar(2^{(1-\varepsilon)\mu(I)})$ algorithm for AR. Separately, we will also show there is no $\ostar(2^{(1-\varepsilon)k})$ algorithm for AR and describe a straight forward $\ostar(2^{k})$ algorithm for MASE (and hence AR).

\begin{lemma}
\label{lemma:complowerbound}
    Assuming SETH, there is no $\ostar(2^{(1-\varepsilon)\mu(I)})$ algorithm for AR.
\end{lemma}
\begin{proof}
    Consider the AR problem where
    \begin{itemize}
        \item $G $ is the Extended Translation of $\varphi(z_1, \ldots z_n)$,
        \item $S = \{ \varphi \} \cup \{ z_i\} \cup \{ \neg z_i \}$, using the same notation as in the definition of the Extended Translation, and
        \item $k = n$.
    \end{itemize}
    By \autoref{ETstructuretheorem}, any non-trivial admissible extension of $G$ contains $\varphi$ and exactly one vertex from each of the $n$ pairs $(z_i, \neg z_i)$ and corresponds to a satisfying assignment for the original CNF. As $k = n$, which is just enough to remove one vertex from each literal pair, this instance of admissible removal is equivalent to finding if there exists a satisfying assignment to $\varphi(z_1, \ldots z_n)$.

    We now note $b = 2n$ (there are $n$ disjoint 2-cycles, one for each variable). Hence $\mu(I) = n$ and an AR algorithm with complexity $\ostar(2^{(1-\varepsilon)\mu(I)})$ would imply an $\ostar(2^{(1-\varepsilon)n})$ algorithm for SAT, contradicting SETH.
\end{proof}

Incidentally, the $2^{\mu(I)}$ bound is also tight for the number of preferred extensions.
\begin{lemma}
    There are instances $I = (G, S, k)$ with $2^{\mu(I)}$ preferred subsets of $S$ that are at distance at most $k$ from $S$.
\end{lemma}
\begin{proof}
    Consider an instance $I=(G,S,k)$, where $G=(V,E)$ is a disjoint union of $n/2$ $2$-cycles, $S=V$, and $k=n/2$. It suffices to note that $\mu(I) = \frac{k}{2} + \frac{b}{4} = \frac{n}{2}$, and that $G$ has $2^{n/2}=2^{\mu(I)}$ preferred subsets, all of which are at distance $k$ from $S$.
\end{proof}

In the same manner we also obtain:
\begin{lemma}
    Assuming SETH, there is no $\ostar(2^{(1-\varepsilon)k})$ algorithm for AR. Furthermore, a simple $\ostar(2^k)$ algorithm exists for MASE (hence AR). In addition, there are at most $2^k$ preferred extensions that are subsets of $S$ with size $\geq \abs{S}-k$.
\end{lemma}
\begin{proof}
    Consider the same instance of AR as in the above proof. Since $k = n$, we again have that an $\ostar(2^{(1-\varepsilon)k})$ algorithm for AR would imply an $\ostar(2^{(1-\varepsilon)n})$ algorithm for SAT, contradicting SETH.

    For existence, we just consider a simple 2-way branching algorithm. While there is an edge, say $u \to v$ in $G[S]$, do a 2-way branch, in 1 branch removing $u$ and in another branch removing $v$. As before, we glue the solutions of our branches together by applying our Maximal Subset Collation algorithm. The base case is when $G[S]$ is conflict-free which can be solved in polynomial time by \autoref{lemma:DAGmaximaladmissible}.

    The enumeration upper bound immediately follows as our branching algorithm enumerates every maximal admissible subset of $S$ (which contains all preferred extensions that are subsets of $S$).
\end{proof}

\section{Discussion of measure}
\label{sect:discussionofmeasure}
We have parameterized by the number of vertices in 2-cycles. It is natural to ask whether we should instead be parameterizing by the number of 2-cycles or some mixture of the two. In this section, we provide some partial results that somewhat justify our choice of measure.

For any instance $I$ of AR, let $b$ be the number of vertices in at least one 2-cycle in $G[S]$ and let $d$ be the number of 2-cycles in $G[S]$. We restrict our consideration to measures of the form $c_1k + c_2b + c_3d$ where $c_1, c_2, c_3$ are real constants. We will prove some results regarding such measures.

Throughout, we will say a measure $\mu$ is valid if the existence of a $\ostar (2^{\mu(I)})$ algorithm for AR does not contradict SETH.

\begin{lemma}
For any valid measure $\mu$, $c_1 \geq \frac{1}{2}$ assuming SETH.
\end{lemma}

\begin{proof}
Let $\varphi(z_1, \ldots z_n)$ be an arbitrary CNF formula. Now we consider the AR problem where
\begin{itemize}
    \item $G$ is the Oriented Translation applied to the Extended Translation of $\varphi(z_1, \ldots z_n)$. Explicitly this is the same construction as the Extended Translation except each literal pair $(z_i, \neg z_i)$ now maps to a 4-cycle $z_{i,1} \to \neg z_{i,1} \to z_{i,2} \to \neg z_{i,2} \to z_{i,1}$,
    \item $S = \{ \varphi \} \cup \{ z_{i,c} : c \in \{1,2\} \} \cup \{ \neg z_{i,c} : c \in \{1,2\} \}$, using the same notation as in our description of $G$, and
\item $k = 2n$.
\end{itemize}

\autoref{orientedtranslationthm} states that the preferred extensions of $G$ are in a bijection with the preferred extensions of the Extended Translation. The preferred extensions of $G$ hence correspond directly to satisfying assignments of $\varphi(z_1, \ldots z_n)$ by \autoref{ETstructuretheorem}. As $k = 2n$ which is just enough to remove two vertices from each literal 4-cycle, this instance of AR reduces to checking if there exists a satisfying assignment to $\varphi(z_1, \ldots z_n)$.

This instance has $b = d = 0$. Hence, if $c_1k < \frac{1}{2}k$, this would imply an $\ostar(2^{(1-\varepsilon)n})$ algorithm for SAT, contradicting SETH.
\end{proof}

\begin{lemma}
For any valid measure $\mu$, $c_1 + 2c_2 + c_3 \geq 1$, assuming SETH.
\end{lemma}

\begin{proof}
Consider the translation from SAT used for our complexity lower bound for MASE/AR on general directed graphs (\autoref{lemma:complowerbound}). This instance has $k = n$, $b = 2n$, $d = n$. Assuming SETH, we must have $c_1n + 2c_2n + c_3n \geq n$ or, equivalently: $c_1 + 2c_2 + c_3 \geq 1$.
\end{proof}

\begin{lemma}
Consider the measures with $c_1 = \frac{1}{2}$. Assuming SETH, among these measures, $\mu(I) = \frac{k}{2} + \frac{b}{4}$ is optimal in the sense that, for any other valid measure $\mu'$ with $c_1 = \frac{1}{2}$, $\mu'(I) \geq \mu(I)$ for all instances $I$.
\end{lemma}

\begin{proof}
Let $I$ be an instance of AR. Let $\mu$ be any valid measure with $c_1 = \frac{1}{2}$. Then,
% TODO WORDING HERE IS POOR.
\begin{flalign*}
      \mu(I) &= \frac{1}{2}k + c_2b + c_3d\\
      \shortintertext{Since the number of vertices in a 2-cycle is at most twice the number of 2-cycles, we have $b \leq 2d$. Hence:}
     &\geq \frac{1}{2}k + c_2b + \frac{c_3}{2}b \\
     \shortintertext{As $c_1 + 2c_2 + c_3 \geq 1$:}
     &\geq \frac{1}{2}k + \frac{1}{4}b
\end{flalign*}
This concludes the proof of the lemma.
\end{proof}
Hence our earlier measure is the optimal valid measure with $c_1 = \frac{1}{2}$, assuming SETH.

We note this says nothing about measures with $c_1 > \frac{1}{2}$. Indeed, the measure $\mu(I) = k$ is better than our new measure on graphs with many vertices in 2-cycles. It also says nothing about other classes of measures. For instance, we believe $\mu(I) = \frac{k}{2} + \min(\frac{k}{2}, \frac{b}{4})$ is a valid measure that is worth exploring.

% TODO: We can show any measure with c_1 + 2c_2 = 1, c_1 >= 1/2 works.
%\openq{Can we extend our results to $c_1 > \frac{1}{2}$. My instinct is this shouldn't be hard and it is likely that for any $\frac{1}{2} \leq c_1 \leq 1$, there is an optimal measure with $c_3 = 0$ and $c_1 + 2c_2 = 1$. And more ambitiously, can we get results on other classes of measures?} I think we can extend it
%
% TODO: INCLUDE KERNEL RESULTS HERE.
%\todo{I also had a result showing impossibility of a linear kernel for defendable removal. It was very straightforward, a kernelization linear to $k$ followed by a translation back to SAT would give an $\abs{V}$ kernel for SAT. Might be worth including, not sure}
% SG: not for the ITCS submission

\section{Detailed Case Analysis of our Improved Enumeration Algorithm}
\label{sect:detailedcaseanalysis}
As a reminder, our cases are:

\begin{center}
	\newcolumntype{L}[1]{>{\hsize=#1\hsize\raggedright\arraybackslash}X}%
    \begin{tabularx}{1\textwidth}{L{0.15} L{1.3}  L{0.8} }
		%\hline
		\toprule
        Case & Requirement to apply & Worst case branching number \\
		\midrule%\hline
        Base & $\Und = \emptyset$. & Solves in $O^{*}(1)$, returns 1 set. \\
		%\hline
        1 & $\exists v \in G[\Und]$ with total degree $\geq 7$. & Branching vector $(8, 1)$, branching number $\varphi \approx 1.2321$. \\
		%\hline
        2 & $\exists v \in G[\Und]$ with degree $(1,-)$. & Branching vector $(4, 3)$, branching number $\approx 1.221$. \\
		%\hline
        3 & $\exists v \in G[\Und]$ with in-degree $\neq$ out-degree. & Branching vector $(6, 5, 5)$, branching number $\approx 1.2298$. \\
		%\hline
        4 & $G[\Und]$ has a weakly connected component where every vertex has degree $(2,2)$. & Branching vector $(6, 5, 5)$, branching number $\approx 1.2298$. \\
		%\hline
        5 & $G[\Und]$ has a weakly connected component where every vertex has degree $(3,3)$. & Branching vector $(7, 7, 7, 7)$, branching number $\approx 1.219$. \\
		%\hline
        6 & There is a weakly connected component in $G[\Und]$ where every vertex has in-degree $=$ out-degree. & Branching vector $(7, 5, 5)$, branching number $\approx 1.218$. \\
		\bottomrule
	\end{tabularx}
\end{center}

And we always apply the first applicable case.

We also recall our simplification rules:

\subsection{Simplification Rules}
Both of these are applicable in polynomial time and decrease $\mu = \abs{\Und}$.

\setcounterref{simpruletwo}{simprule:outdeg0}
\addtocounter{simpruletwo}{-1}

\begin{simpruletwo}[Out-degree 0]
\label{simprule:appendix_outdeg0}
Let $v$ be a vertex in $G[\Und]$ with out-degree $0$. Move $v$ from $\Und$ to the end of the queue $\Def$.
\end{simpruletwo}

\setcounterref{simpruletwo}{simprule:indeg0}
\addtocounter{simpruletwo}{-1}

\begin{simpruletwo}[In-degree 0]
\label{simprule:appendix_indeg0}
Let $v$ be any vertex in $G[\Und]$ with in-degree $0$. Then by Invariant 3, $v$ has in-degree $0$ in $G[\Und \cup \Def]$.
Applying Simplification Rule \textbf{(Undefendable)}
we can set $\Und \gets \Und \setminus N(v)$. After that, $v$ has out-degree 0 in $G[\Und]$ and hence we move $v$ from $\Und$ to $\Def$.
Our new instance $I' = (\Und', \Def')$ has:
\begin{itemize}
    \item $\Und' = \Und \setminus N[v]$.
    \item $\Def' = \Def \cup \{v\}$.
\end{itemize}
\end{simpruletwo}

Due to these rules, henceforth we may assume each vertex in $G[\Und]$ has in-degree $\geq 1$, out-degree $\geq 1$ and (total) degree $\geq 2$.

\subsection{Notation}
\label{subsect:Oriented_Notation}
As in \autoref{subsect:Oriented_ExtraNotation} we will say a vertex has degree $(a,-)$ if it has in-degree $a$, a vertex has degree $(-,b)$ if it has out-degree $b$ and a vertex has degree $(a,b)$ if it has in-degree $a$ and out-degree $b$.

We will again use the shorthand \emph{2-way branch on whether to include $v$} (see \autoref{subsect:Oriented_2waybranch}).

Similarly, we will use the phrase \emph{include $v$} as shorthand for the first part of the branching rule described in \autoref{subsect:Oriented_2waybranch}.
This terminology will be useful for describing our 3-way and 4-way branches (these will be similar to our 2-way branches except there will be a choice on which vertex to include from a specified set).

% TODO: remove any duplicate information here?
\subsection{Base Case - \texorpdfstring{$\Und = \emptyset$}{Und is empty}}
\label{subsect:Oriented_BaseCase}

We need to enumerate all maximal admissible subsets of $\Und \cup \Def = \Def$. By Invariant \ref{invariant:2}, $G[\Def]$ is a DAG. Hence by \autoref{lemma:DAGmaximaladmissible} there is exactly one maximal admissible subset of $\Def$ and we can find it in polynomial time.

\subsection{Case 1 - \texorpdfstring{Total Degree $\geq 7$}{Total degree at least 7 exists}}
\label{subsect:Oriented_Case1}

Let $v$ be any vertex in $G[\Und]$ with total degree $\geq 7$.
We do a \branchon{v}.

This gives a branching vector at least as good as $(8,1)$ with branching number at most $\varphi$ (recall that $\varphi$ was defined to be the branching number for $(8,1)$).

\subsection{Case 2 - \texorpdfstring{Degree $(1,-)$ exists}{Degree (1,-) exists}}
\label{subsect:Oriented_Case2}

Let $v$ be any vertex in $G[\Und]$ with degree $(1,-)$. Let $a$ be the vertex attacking $v$ and let $b$ be any vertex $v$ attacks (due to \autoref{simprule:appendix_outdeg0}, $v$ has out-degree $>0$).
We note that $a \neq b$ as $G$ is oriented.
See \autoref{fig:Oriented_Case2} for a depiction.

We consider a few cases depending on the degrees of $a$, $b$, and $v$.

\subsubsection*{Case 2.1 - $\deg(a) \geq 3$ or $\deg(v) \geq 3$}

% TODO: check this is in right place after.
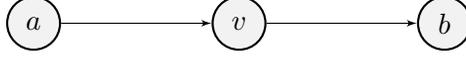
\begin{figure}[tb]
	\centering
	\begin{tikzpicture}
	\tikzstyle{nd}=[circle,draw=black,thick,fill=gray!10,minimum size=20pt,inner sep=0pt]
	\node[nd](a) at (0,0) {$a$};
	\node[nd, right = 2cm of a](v) {$v$};
	\node[nd, right = 2cm of v](b) {$b$};
	\begin{scope}
		\tikzset{edge/.style = {->,> = latex'}}
        \draw[edge] (a) -- (v);
        \draw[edge] (v) -- (b);
	\end{scope}
	\end{tikzpicture}
    \caption{Key vertices in Case 2}
    \label{fig:Oriented_Case2}
\end{figure}

If either $a$ or $v$ has degree at least $3$, then we do a \branchon{a}. This gives a branching vector of $(\deg(a)+1, 1)$.

Furthermore, in the branch where we exclude $a$, $v$ is left with in-degree $0$. Hence, we can apply \autoref{simprule:appendix_indeg0} to $v$.
Therefore, the branch that excludes $a$ also removes all of $N[v]$ from $\Und$.

This is a branching rule with branching vector $(\deg(a)+1, \deg(v)+1)$. After our simplification rules, all vertices in $G[\Und]$ have total degree at least $2$.
By assumption at least one of $a$ or $v$ has degree at least 3. Hence, our branching vector is at least as good as $(4,3)$.

\subsubsection*{Case 2.2 - $\deg(b) \geq 3$}
Vertices $a$ and $v$ do not have in-degree or out-degree $0$ due to our simplification rules. Hence, they must have degree $(1,1)$ (else the previous case is applicable).

If $\deg(b) \geq 3$, we do a \branchon{b}.

In the branch where $b$ is excluded $v$ will have out-degree 0. Hence we may apply \autoref{simprule:appendix_outdeg0} to $v$. After this, $a$ will also have out-degree 0 and we may apply \autoref{simprule:appendix_outdeg0} again to $a$.
Hence our new instance $I' = (\Und', \Def')$ has:
    \begin{itemize}
        \item $\Und' = \Und \setminus \{a,v,b\}$.
        \item $\Def' = \Def \cup \{v,a\}$.
    \end{itemize}
    and hence $\mu(I') = \mu(I) - 3$.

Hence we have branching vector $(\deg(b)+1, 3)$, which is at least as good as $(4,3)$.

\subsubsection*{Case 2.3 - $\deg(a) = \deg(v) = \deg(b) = 2$}
Due to our simplification rules, this is the only case remaining. If these three vertices are connected to no other vertices, then no preferred extension can contain any vertex in $\{a,v,b\}$ as $\{a,v,b\}$ is an isolated odd length cycle. Hence, we can remove all three vertices from $\Und$.

% TODO: NTS: Overkilled above, only need at least 1 neighbor to have degree 2. Though may need to act on the vertex with a neighbor with degree > 2 if one exists in the chain, anyways :/
Otherwise, the vertex that $b$ attacks is distinct from $\{a,v,b\}$. Let this vertex be $w$. We have a chain of attacks $a \to v \to b \to w$. Do a \branchon{a}.

In the branch where we include $a$, since $v \in N[a]$, $v$ is removed from $\Und$. This leaves $b$ with in-degree 0. Hence we may apply \autoref{simprule:appendix_indeg0} to $b$.
Our new instance $I' = (\Und', \Def')$ has:
\begin{itemize}
    \item $\Und' = \Und \setminus \{a,v,b,w\}$.
    \item $\Def' = \Def \cup \{a,b\}$.
\end{itemize}
and hence $\mu(I') = \mu(I) - 4$.

In the branch where we exclude $a$, $v$ has in-degree 0. We may apply \autoref{simprule:appendix_indeg0} to $v$ removing $b$ from $\Und$.
Hence our new instance $I' = (\Und', \Def')$ has:
\begin{itemize}
    \item $\Und' = \Und \setminus \{a,v,b\}$.
    \item $\Def' = \Def \cup \{v\}$.
\end{itemize}
with $\mu(I') = \mu(I) - 3$.

Hence we have branching vector $(4,3)$.

\subsubsection*{Case 2 - Summary}
Hence there is always a branching rule with branching vector $(4,3)$ which has branching number $\approx 1.221 < \varphi$.

After this case, all vertices can be assumed to have in-degree $\geq 2$, out-degree $\geq 1$ and total degree $\geq 3$.
%\todo{I think this case can be improved to $(5,3)$ with a lot more case work. Not necessary for overall result but might be nice to note}

\subsection{Case 3 - \texorpdfstring{Vertex with in-degree $\neq$ out-degree}{Vertex with in-degree not equal to out-degree exists}}
\label{subsect:Oriented_Case3}
As the sum of all in-degrees is equal to the sum of all out-degrees in $G[S]$, there exists a vertex $v$ in $G[S]$ with larger out-degree than in-degree.

Due to the above cases, $v$ has in-degree $\geq 2$ and total degree $\leq 6$. Hence, $v$ has degree $(2,3)$ or $(2,4)$. There are two cases depending on the vertices attacking $v$.

\subsubsection*{Case 3.1 - A degree-3 vertex attacks $v$}
If either of the vertices attacking $v$ has total degree $3$, then it must have degree $(2,1)$ due to the previous cases. Let $u$ be such a vertex.

We do a \branchon{v}. This gives a branching vector of $(\deg(v)+1, 1)$.

Furthermore, in the branch where we exclude $v$, $u$ will have out-degree 0. Hence we may apply \autoref{simprule:appendix_outdeg0} to move $u$ from $\Und$ to $\Def$.

Hence our branching rule has a branching vector at least as good as $(6,2)$, with branching number $\approx 1.211$.

\subsubsection*{Case 3.2 - Both in-neighbors of $v$ have degree  at least $4$}
Otherwise, both vertices attacking $v$ have total degree at least $4$.
Let the 2 vertices be $a$ and $b$.

We do a 3-way branch, in the first branch we \emph{include $a$}, in the second branch we \emph{include $b$} and in the third branch we exclude both $a$ and $b$.
The branches including $a$ and $b$ remove at least 5 vertices from $\Und$ as $\deg(a), \deg(b) \geq 4$.

The branch excluding both $a$ and $b$ results in $v$ having in-degree 0. Hence we may apply \autoref{simprule:appendix_outdeg0} to $v$. Our new instance $I' = (\Und', \Def')$ then has:
\begin{itemize}
    \item $\Und' = \Und \setminus N[v]$.
    \item $\Def' = \Def \cup \{v\}$.
\end{itemize}
with $\mu(I') \leq \mu(I) - 6$.

This is a $(6,5,5)$ branching with branching number $\approx 1.2298$.
% Note: This is one of the second worst cases. Worst is a degree 7 vertex exists}

\subsection{Case 4 - \texorpdfstring{$G[\Und]$ contains a weakly connected component of degree $(2,2)$ vertices}{G[Und] contains a weakly connected component of degree (2,2) vertices}}
\label{subsect:Oriented_Case4}
Let $C$ be a weakly connected component in $G[\Und]$ such that each $v \in C$ has degree $(2,2)$. We consider two cases.

\subsubsection*{Case 4.1 - There exist non-adjacent attackers}
In this subcase we assume there exists a vertex $v \in C$ such that its two in-neighbors $a, b \in C$ are not adjacent. Then we do a 3-way branch: \emph{include $a$}, exclude $a$ but \emph{include $b$}, exclude both.

In the branch where we \emph{include $a$} our new instance $I' = (\Und', \Def')$ has:
\begin{itemize}
    \item $\Und' = \Und \setminus N[a]$.
    \item $\Def' = \Def \cup \{a\}$.
\end{itemize}
with $\mu(I') = \mu(I) - 5$.

In the branch where we \emph{include $b$} and exclude $a$ our new instance $I' = (\Und', \Def')$ has:
\begin{itemize}
    \item $\Und' = \Und \setminus (N[b] \sqcup \{a\})$.
    \item $\Def' = \Def \cup \{b\}$.
\end{itemize}
with $\mu(I') = \mu(I) - 6$.

In the branch where we exclude $a$ and $b$, $v$ has in-degree 0.

Hence we may apply \autoref{simprule:appendix_outdeg0} to $v$. Our new instance $I' = (\Und', \Def')$ then has:
\begin{itemize}
    \item $\Und' = \Und \setminus N[v]$.
    \item $\Def' = \Def \cup \{v\}$.
\end{itemize}
with $\mu(I') = \mu(I) - 5$.

This is a $(6,5,5)$ branching with branching number $\approx 1.2298$. % Note: This is one of the second worst cases. Worst is a degree 7 vertex exists}

\subsubsection*{Case 4.2 - Attackers of each vertex are adjacent}
In this subcase, each $v \in C$ has degree $(2,2)$ with adjacent in-neighbors. We handle this case by showing that the graph $G[C]$ has a very restricted structure. In particular, we show that $G[C]$ must look like \autoref{2regulargraph_example}.

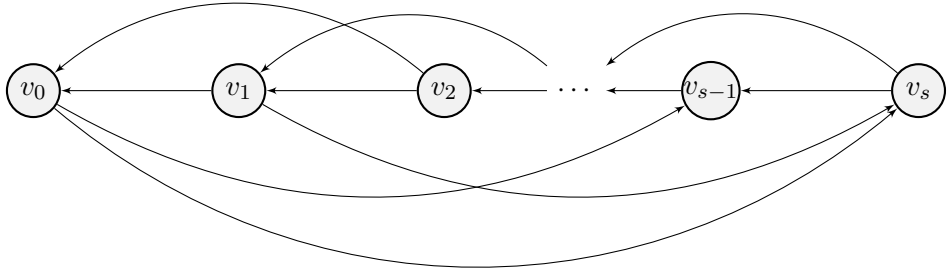
\begin{figure}[tb]
	\centering
	\begin{tikzpicture}
	\tikzstyle{nd}=[circle,draw=black,thick,fill=gray!10,minimum size=20pt,inner sep=0pt]
	\node[nd](v0) at (0,0) {$v_0$};
	\node[nd, right = 2cm of v0](v1) {$v_1$};
	\node[nd, right = 2cm of v1](v2) {$v_2$};
    \node[minimum size=20pt, right = 1cm of v2](vdots) {$\cdots$};
    \node[nd, right = 1cm of vdots](vn-1) {$v_{s-1}$};
    \node[nd, right = 2cm of vn-1](vn) {$v_s$};
	\begin{scope}
		\tikzset{edge/.style = {->,> = latex'}}
		\draw[edge] (v1) -- (v0);
		\draw[edge] (v2) -- (v1);
        \draw[edge] (v2) to [out=140, in=40] (v0);
		\draw[edge] (vdots) -- (v2);
		\draw[edge] (vn-1) -- (vdots);
        \draw[edge] (vdots) to [out=140, in=40] (v1);
        \draw[edge] (vn) to [out=140, in=40] (vdots);
		\draw[edge] (vn) -- (vn-1);
        \draw[edge] (v0) to [out=-30, in=210] (vn-1);
        \draw[edge] (v0) to [out=-40, in=220] (vn);
        \draw[edge] (v1) to [out=-30, in=210] (vn);
	\end{scope}
	\end{tikzpicture}
    \caption{Example $F_n$ for \autoref{appendix:struct2regular}}
    \label{2regulargraph_example}
\end{figure}

\begin{theorem}
	\label{appendix:struct2regular}
	% TODO: is orientedness necessary?
	Let $G = (V,E)$ be a weakly connected, oriented graph that satisfies the following properties:
	\begin{enumerate}
        \item [P1] every vertex has in-degree 2 and out-degree 2, and
        \item [P2] the in-neighbors of every vertex are adjacent.
	\end{enumerate}
	Then $\abs{V} \geq 5$ and $G$ is isomorphic to $F_{\abs{V}}$ where we define $F_n$ for $n \in \mathbb{N}$ as:
	\begin{itemize}
		\item $V(F_n) = \{0,1,\ldots,n-1\}$
		\item $E(F_n) = \bigcup\limits_{i=0}^{n-1} \edge{i,(i+1)\mod n} \cup \bigcup\limits_{i=0}^{n-1} \edge{i,(i+2) \mod n}$
	\end{itemize}
\end{theorem}

\begin{proof}
	Consider any arbitrary vertices $v_0, v_1, v_2$ with $\edge{v_1, v_0}, \edge{v_2, v_1}, \edge{v_2, v_0} \in E$. Such a triplet exists as we can arbitrarily pick any $v_0$ and then pick $v_1$ and $v_2$ to be its in-neighbors such that the edge between $v_1$ and $v_2$ is $\edge{v_2, v_1}$. Let $V_2 = [v_0, v_1, v_2]$.  Then, while $v_{k-1}$ has only one in-neighbor in $G[V_k]$, we construct $V_{k+1}$ to be $V_k$ appended with $v_{k+1}$ where $v_{k+1}$ is the other in-neighbor of $v_{k-1}$. 
	
	\begin{lemma}
		Each $V_k$ satisfies the following property:
		\begin{center}
			\textbf{Local Attack Property:} For all $i \geq 1$, $v_i$ attacks $v_{i-1}$ and for all $i \geq 2$, $v_i$ attacks $v_{i-2}$.
		\end{center}
	\end{lemma}
	
	\begin{proof}
		This is obvious for the base case $V_2$. We can then confirm it for $V_k, k \geq 3$ by induction:
		
		The vertex $v_k$ attacks $v_{k-2}$ by construction.
		
		Hence, by \textit{(P2)}, either $\edge{v_k, v_{k-1}} \in E$ or $\edge{v_{k-1}, v_k} \in E$. However, since $V_{k-1}$ satisfies the Local Attack Property and $k-1 \geq 2$, $v_{k-1}$ attacks both $v_{k-2}$ and $v_{k-3}$. Hence, it already has two out-neighbors distinct from $v_k$. Therefore, since $v_{k-1}$ has out-degree 2 by \textit{(P1)}, it must be the case that $\edge{v_k,v_{k-1}} \in E$ as required.
	\end{proof}
	
	Now, consider the largest such constructed sequence, $V_s$. Then, since $V_s$ could not be extended, there must be a vertex $u \in V_s, u \neq v_s$ that attacks $v_{s-1}$. Since $V_s$ satisfies the Local Attack Property, every vertex other than $v_s$, $v_1$ and $v_0$ already has 2 out-neighbors that are not $v_{s-1}$. Hence $u = v_0$ or $u = v_1$. We now note that since $G$ is oriented, for $v_0$ or $v_1$ to have an edge to $v_{s-1}$, $s$ must be at least 4.
	
	Now, as $s \geq 4$, $v_s$ attacks two vertices that are not $u$. Hence, by \textit{(P2)}, $u$ must also attack $v_s$. Therefore, by out-degree considerations, $u = v_0$. Finally, by \textit{(P1)}, the other in-neighbor of $v_s$ must attack $v_0$. By simple degree considerations, we conclude that the other in-neighbor of $v_s$ is $v_1$. Hence, $\edge{v_0, v_{s-1}}, \edge{v_0, v_s}, \edge{v_1, v_s} \in E$ and $V_s$ actually satisfies the following stronger condition:
	\begin{center}
		For all $i \leq s$, $\edge{v_i, v_{(i-1)\mod s}}, \edge{v_i, v_{(i-2)\mod s}} \in E$.
	\end{center}
    Therefore, $G[V_s]$ contains a subgraph isomorphic to $F_{s+1}$ with an obvious isomorphism. Furthermore, by degree considerations and $\textit{(P1)}$, this subgraph is $G[V_s]$ itself and $G[V_s]$ is isomorphic to $F_{s+1}$. Finally, since $G$ is weakly connected and $G[V_s]$ satisfies \textit{(P1)}, we conclude that $V_s$ contains all of $V$ and $s+1 = \abs{V}$. Hence, since $s \geq 4$, we further conclude that $\abs{V} \geq 5$.
\end{proof}

\autoref{appendix:struct2regular} shows that our weakly connected component $C$ is isomorphic to $F_n$ for some $n \geq 5$.

We do a 3-way branch on \emph{including $v_0$}, \emph{including $v_1$}, and excluding both.
In each of these cases $v_0$ and $v_1$ are both removed from $\Und$ (as $v_0$ is adjacent to $v_1$) and it is easy to confirm that the resulting graph, $C' = C \setminus \{v_0,v_1\}$, is a DAG.
We can repeatedly apply \autoref{simprule:appendix_outdeg0} to vertices with out-degree $0$ in the DAG $C'$ until all vertices in $C'$ have been moved to $\Def$.
Hence in each of these branches, once we finish applying our simplification rules, the size of $\Und$ decreases by $\abs{C}$.

We now note that $F_5$ has no non-trivial admissible extensions. This also follows from a more general argument: if $v_i$ is in an admissible extension that is a subset of $\Und \cup \Def$, then so is $v_{(i + 3k) \mod n}$ for all $k$. Hence if $3 \nmid n$, then $F_n$ has only the trivial admissible extension.

Hence we may assume $n \geq 6$. Hence our 3-way branching rule has branching vector at least as good as $(6,6,6)$ with branching number $3^{\frac{1}{6}} < \varphi$.
%\openq{Would be interesting to know if there is a branching argument for this case, rather than this graph structure argument}

\subsection{Case 5 - \texorpdfstring{$G[\Und]$ contains a weakly connected component of degree $(3,3)$ vertices}{G[Und] contains a weakly connected component of degree (3,3) vertices}}
\label{subsect:Oriented_Case5}
Let $C$ be a weakly connected component in $G[\Und]$ such that each $v \in C$ has degree $(3,3)$. Pick any vertex $v \in C$ and let its 3 attackers be $\{a,b,c\}$. Then we do a 4-way branch: \emph{include $a$}, \emph{include $b$}, \emph{include $c$}, and exclude all of $\{a,b,c\}$.

By assumption, $a,b,c$ all have degree $(3,3)$. Hence the first 3 cases each remove at least 7 vertices from $\Und$. In the last case, $v$ has in-degree 0 and hence by \autoref{simprule:appendix_indeg0}, $N[v]$ is removed from $\Und$.

Hence, this is a $(7,7,7,7)$ branching with branching number $4^{\frac{1}{7}} \approx 1.219$.

\subsection{Case 6 - Every vertex has the same in-degree as out-degree}
\label{subsect:Oriented_Case6}
In this case, every vertex has the same in-degree as out-degree.
From \hyperref[subsect:Oriented_Case2]{case 2} no vertex has in-degree $\leq 1$.
From \hyperref[subsect:Oriented_Case1]{case 1} no vertex has total degree $\geq 7$.
Hence, each vertex has degree $(2,2)$ or $(3,3)$.

Furthermore, due to the previous two cases, each weakly connected component $C$ of $G[\Und]$ contains both a vertex with degree $(2,2)$ and a vertex with degree $(3,3)$.

Hence we may apply the following lemma:
\begin{lemma}
    Let $G = (V,E)$ be a weakly connected, oriented graph where every vertex $v \in G$ has degree $(2,2)$ or $(3,3)$.
    Further suppose $G$ contains at least one vertex with degree $(2,2)$ and at least one vertex with degree $(3,3)$.

    Then there exists a $v \in G$ with degree $(2,2)$ that is attacked by a vertex $a \in G$ with degree $(3,3)$.
\end{lemma}
\begin{proof}
    Let $V_2 \subseteq V$ be the vertices with degree $(2,2)$. Let $V_3 = V \setminus V_2$ be the vertices with degree $(3,3)$.

    Accounting for the edges in $G[V_2]$ and the edges between $V_2$ and $V_3$ separately, we have:
    \[
        \text{$\abs{E(G[V_2])} + $ (number of edges from $V_3$ to $V_2$)} = \sum_{v \in V_2} \mathrm{indegree}(v) = 2\abs{V_2}
    \]
    Similarly, we have:
    \[
        \text{$\abs{E(G[V_2])} + $ (number of edges from $V_2$ to $V_3$)} = \sum_{v \in V_2} \mathrm{outdegree}(v) = 2\abs{V_2}
    \]
    Hence:
    \begin{center}
        number of edges from $V_3$ to $V_2$ = number of edges from $V_2$ to $V_3$
    \end{center}
    As $G$ is weakly connected and neither $V_2$ nor $V_3$ are empty, there exists at least one edge between $V_2$ and $V_3$.

    Hence there exists an edge from $V_3$ to $V_2$ as required.
\end{proof}

Let $v$ and $a$ be as in the above lemma and let $b$ be $v$'s other attacker.

Now we do a 3-way branch.
\begin{itemize}
    \item \emph{Include $a$}. Then we remove $\abs{N[a]} = 7$ vertices from $\Und$.
    \item \emph{Include $b$}. Then we remove $\abs{N[b]} \geq 5$ vertices from $\Und$.
	\item Exclude $a$ and $b$. Then $v$ has in-degree $0$ and applying \autoref{simprule:appendix_indeg0} removes $\abs{N[v]} = 5$ vertices from $\Und$.
\end{itemize}

This is a $(7,5,5)$ branching rule with branching number $\approx 1.218$.

\end{appendices}

\bibliographystyle{plainurl}
\bibliography{pub}

\begin{thebibliography}{10}

\bibitem{preferencebasedafs}
Leila Amgoud and Claudette Cayrol.
\newblock A reasoning model based on the production of acceptable arguments.
\newblock {\em Annals of Mathematics and Artificial Intelligence},
  34(1-3):197--215, 2002.
\newblock URL: \url{https://doi.org/10.1023/A:1014490210693}, \href
  {http://dx.doi.org/10.1023/A:1014490210693}
  {\path{doi:10.1023/A:1014490210693}}.

\bibitem{BanderierLR04}
Cyril Banderier, Jean-Marie~Le Bars, and Vlady Ravelomanana.
\newblock Generating functions for kernels of digraphs (enumeration \&
  asymptotics for a constraint from game theory).
\newblock In {\em Proceedings of the 16th International Conference on Formal
  Power Series and Algebraic Combinatorics (FPSAC 2004)}, pages 91--105, 2004.

\bibitem{BaumannS13}
Ringo Baumann and Hannes Strass.
\newblock On the maximal and average numbers of stable extensions.
\newblock In {\em Proceedings of the 2nd International Workshop on Theory and
  Applications of Formal Argumentation (TAFA 2013)}, volume 8306 of {\em
  Lecture Notes in Computer Science}, pages 111--126. Springer, 2013.

\bibitem{BaumannS14a}
Ringo Baumann and Hannes Strass.
\newblock Open problems in abstract argumentation.
\newblock In {\em Advances in Knowledge Representation, Logic Programming, and
  Abstract Argumentation - Essays Dedicated to Gerhard Brewka on the Occasion
  of His 60th Birthday}, volume 9060 of {\em Lecture Notes in Computer
  Science}, pages 325--339. Springer, 2015.

\bibitem{valuebasedafs}
Trevor J.~M. Bench{-}Capon.
\newblock Value based argumentation frameworks.
\newblock In {\em Proceedings of the 9th International Workshop on
  Non-Monotonic Reasoning (NMR 2002)}, volume cs.AI/0207059, pages 443--454,
  2002.
\newblock URL: \url{http://arxiv.org/abs/cs.AI/0207059}.

\bibitem{Bisdorff06}
Raymond Bisdorff.
\newblock On enumerating the kernels in a bipolar-valued outranking digraph.
\newblock Technical Report~6, Annales du Lamsade, 2006.
\newblock hal-00118995.

\bibitem{BistarelliRS15}
Stefano Bistarelli, Fabio Rossi, and Francesco Santini.
\newblock A comparative test on the enumeration of extensions in abstract
  argumentation.
\newblock {\em Fundamenta Informaticae}, 140(3-4):263--278, 2015.

\bibitem{Caminada07}
Martin Caminada.
\newblock An algorithm for computing semi-stable semantics.
\newblock In {\em Proceedings of the 9th European Conference on Symbolic and
  Quantitative Approaches to Reasoning with Uncertainty (ECSQARU 2007)}, volume
  4724 of {\em Lecture Notes in Computer Science}, pages 222--234. Springer,
  2007.

\bibitem{semistablealgo}
Martin Caminada.
\newblock An algorithm for computing semi-stable semantics.
\newblock In {\em ECSQARU 2007: Symbolic and Quantitative Approaches to
  Reasoning with Uncertainty}, volume 4724 of {\em Lecture Notes in Computer
  Science}, pages 222--234. Springer, Berlin, Heidelberg, 2007.

\bibitem{CeruttiDGV13}
Federico Cerutti, Paul~E. Dunne, Massimiliano Giacomin, and Mauro Vallati.
\newblock Computing preferred extensions in abstract argumentation: {A}
  sat-based approach.
\newblock In {\em Proceedings of the 2nd International Workshop on Theory and
  Applications of Formal Argumentation (TAFA 2013)}, volume 8306 of {\em
  Lecture Notes in Computer Science}, pages 176--193. Springer, 2013.

\bibitem{CeruttiGV14}
Federico Cerutti, Massimiliano Giacomin, and Mauro Vallati.
\newblock Algorithm selection for preferred extensions enumeration.
\newblock In {\em Proceedings of the 5th International Conference on
  Computational Models of Argument (COMMA 2014)}, volume 266 of {\em Frontiers
  in Artificial Intelligence and Applications}, pages 221--232. {IOS} Press,
  2014.

\bibitem{CeruttiVG18}
Federico Cerutti, Mauro Vallati, and Massimiliano Giacomin.
\newblock On the impact of configuration on abstract argumentation automated
  reasoning.
\newblock {\em International Journal of Approximate Reasoning}, 92:120--138,
  2018.

\bibitem{CharwatDGWW15}
G{\"{u}}nther Charwat, Wolfgang Dvor{\'{a}}k, Sarah~Alice Gaggl, Johannes~Peter
  Wallner, and Stefan Woltran.
\newblock Methods for solving reasoning problems in abstract argumentation -
  {A} survey.
\newblock {\em Artificial Intelligence}, 220:28--63, 2015.

\bibitem{ParameterizedAlgorithms}
Marek Cygan, Fedor~V. Fomin, Łukasz Kowalik, Daniel Lokshtanov, Dániel Marx,
  Marcin Pilipczuk, Michał Pilipczuk, and Saket Saurabh.
\newblock {\em Parameterized Algorithms}.
\newblock Springer, 2015.

\bibitem{graphtheoreticalstructures}
Yannis Dimopoulos and Albert Torres.
\newblock Graph theoretical structures in logic programs and default theories.
\newblock {\em Theoretical Computer Science}, 170:209--224, 1996.

\bibitem{preferredextensionsquery}
Sylvie Doutre and Jérôme Mengin.
\newblock Preferred extensions of argumentation frameworks: Query, answering,
  and computation.
\newblock In {\em Proceedings of the 1st International Joint Conference on
  Automated Reasoning (IJCAR 2001)}, volume 2083 of {\em Lecture Notes in
  Computer Science}, pages 272--288. Springer, 2001.

\bibitem{dung}
Phan~Minh Dung.
\newblock On the acceptability of arguments and its fundamental role in
  non-monotonic reasoning, logic programming and n-person games.
\newblock {\em Artificial Intelligence}, 77:321--357, 1995.

\bibitem{DunneDLW15}
Paul~E. Dunne, Wolfgang Dvor{\'{a}}k, Thomas Linsbichler, and Stefan Woltran.
\newblock Characteristics of multiple viewpoints in abstract argumentation.
\newblock {\em Artificial Intelligence}, 228:153--178, 2015.

\bibitem{multipleviewpoints}
Paul~E. Dunne, Wolfgang Dvořák, Thomas Linsbichler, and Stefan Woltran.
\newblock Characteristics of multiple viewpoints in abstract argumentation.
\newblock {\em Artificial Intelligence}, 228:153--178, 2015.

\bibitem{ArgumentationBookComplexityCh}
Paul~E. Dunne and Michael Wooldridge.
\newblock Complexity of abstract argumentation.
\newblock In Iyad Rahwan and Guillermo~R. Simari, editors, {\em Argumentation
  in Artificial Intelligence}, chapter~5, pages 85--104. Springer, Boston, MA,
  2009.

\bibitem{monotonelocalsearch}
Fedor~V. Fomin, Serge Gaspers, Daniel Lokshtanov, and Saket Saurabh.
\newblock Exact algorithms via monotone local search.
\newblock In {\em Proceedings of the 48th Annual {ACM} {SIGACT} Symposium on
  Theory of Computing ({STOC} 2016)}, pages 764--775. {ACM}, 2016.

\bibitem{FominGK09}
Fedor~V. Fomin, Fabrizio Grandoni, and Dieter Kratsch.
\newblock A measure {\&} conquer approach for the analysis of exact algorithms.
\newblock {\em Journal of the {ACM}}, 56(5):25:1--25:32, 2009.

\bibitem{exactexpalgos}
Fedor~V. Fomin and Dieter Kratsch.
\newblock {\em Exact Exponential Algorithms}.
\newblock Springer, 2010.

\bibitem{Galeana-SanchezL98}
Hortensia Galeana{-}S{\'{a}}nchez and Xueliang Li.
\newblock Semikernels and (\emph{k, l})-kernels in digraphs.
\newblock {\em {SIAM} Journal on Discrete Mathematics}, 11(2):340--346, 1998.

\bibitem{Galeana-SanchezN84}
Hortensia Galeana{-}S{\'{a}}nchez and Victor Neumann{-}Lara.
\newblock On kernels and semikernels of digraphs.
\newblock {\em Discrete Mathematics}, 48(1):67--76, 1984.

\bibitem{sergethesis}
Serge Gaspers.
\newblock {\em Exponential time algorithms: Structures, measures, and bounds}.
\newblock PhD thesis, University of Bergen, 2008.

\bibitem{multivariatesubroutines}
Serge Gaspers and Edward~J. Lee.
\newblock {Exact Algorithms via Multivariate Subroutines}.
\newblock In {\em Proceedings of the 44th International Colloquium on Automata,
  Languages, and Programming (ICALP 2017)}, volume~80 of {\em Leibniz
  International Proceedings in Informatics (LIPIcs)}, pages 69:1--69:13.
  Schloss Dagstuhl--Leibniz-Zentrum fuer Informatik, 2017.
\newblock URL: \url{http://drops.dagstuhl.de/opus/volltexte/2017/7425}, \href
  {http://dx.doi.org/10.4230/LIPIcs.ICALP.2017.69}
  {\path{doi:10.4230/LIPIcs.ICALP.2017.69}}.

\bibitem{KrollPW17}
Markus Kr{\"{o}}ll, Reinhard Pichler, and Stefan Woltran.
\newblock On the complexity of enumerating the extensions of abstract
  argumentation frameworks.
\newblock In {\em Proceedings of the 26th International Joint Conference on
  Artificial Intelligence ({IJCAI} 2017)}, pages 1145--1152. ijcai.org, 2017.

\bibitem{hierarchicalargumentation}
Sanjay Modgil.
\newblock Hierarchical argumentation.
\newblock In {\em Proceedings of the 10th European Conference on Logics in
  Artificial Intelligence ({JELIA} 2006)}, pages 319--332, 2006.
\newblock URL: \url{https://doi.org/10.1007/11853886_27}, \href
  {http://dx.doi.org/10.1007/11853886_27} {\path{doi:10.1007/11853886_27}}.

\bibitem{ModgilC09}
Sanjay Modgil and Martin Caminada.
\newblock Proof theories and algorithms for abstract argumentation frameworks.
\newblock In {\em Argumentation in Artificial Intelligence}, pages 105--129.
  Springer, 2009.

\bibitem{ModgilP12}
Sanjay Modgil and Henry Prakken.
\newblock Resolutions in structured argumentation.
\newblock In {\em Proceedings of the 4th International Conference on
  Computational Models of Argument (COMMA 2012)}, volume 245 of {\em Frontiers
  in Artificial Intelligence and Applications}, pages 310--321. {IOS} Press,
  2012.

\bibitem{MoonM65}
John~W. Moon and Leo Moser.
\newblock On cliques in graphs.
\newblock {\em Israel Journal of Mathematics}, 3:23--28, 1965.

\bibitem{NeumannLara71}
Victor Neumann-Lara.
\newblock Semin{\'u}cleos de una digr{\'a}fica.
\newblock Technical report, Anales del Instituto de Matem{\'a}ticas II,
  Universidad Nacional Aut{\'o}noma M{\'e}xico, 1971.

\bibitem{NofalAD14}
Samer Nofal, Katie Atkinson, and Paul~E. Dunne.
\newblock Algorithms for decision problems in argument systems under preferred
  semantics.
\newblock {\em Artificial Intelligence}, 207:23--51, 2014.

\bibitem{SzwarcfiterC94}
Jayme~Luiz Szwarcfiter and Guy Chaty.
\newblock Enumerating the kernels of a directed graph with no odd circuits.
\newblock {\em Information Processing Letters}, 51(3):149--153, 1994.

\bibitem{VallatiCG14a}
Mauro Vallati, Federico Cerutti, and Massimiliano Giacomin.
\newblock Argumentation extensions enumeration as a constraint satisfaction
  problem: a performance overview.
\newblock In {\em Proceedings of the International Workshop on Defeasible and
  Ampliative Reasoning (DARe@ECAI 2014)}, volume 1212 of {\em {CEUR} Workshop
  Proceedings}. CEUR-WS.org, 2014.

\bibitem{VallatiCG14}
Mauro Vallati, Federico Cerutti, and Massimiliano Giacomin.
\newblock Argumentation frameworks features: an initial study.
\newblock In {\em Proceedings of the 21st European Conference on Artificial
  Intelligence (ECAI 2014)}, volume 263 of {\em Frontiers in Artificial
  Intelligence and Applications}, pages 1117--1118. {IOS} Press, 2014.

\bibitem{NeumannM44}
John von Neumann and Oskar Morgenstern.
\newblock {\em Theory of Games and Economic Behavior}.
\newblock Princeton University Press, 1944.

\end{thebibliography}

\end{document}